\documentclass[11pt]{amsart}
\usepackage[usenames,dvipsnames,svgnames,table]{xcolor} 
\usepackage{amssymb,amsfonts,amsrefs, eufrak}
\usepackage{amsmath,amsthm, marvosym}

\usepackage{eulervm}         
\usepackage{graphicx, subfigure}

\usepackage{pgf, tikz} 
\usepackage[active]{srcltx}		
\usepackage{pdfsync}					
\usepackage{enumitem}
\usepackage[colorinlistoftodos]{todonotes}
\usepackage{dsfont}
\usepackage[pdfdisplaydoctitle,colorlinks,urlcolor=blue,linkcolor=blue,citecolor=blue]{hyperref}

\usepackage{boldline}

\newcommand{\dd}{\,{\rm d}}

\newcommand{\Tr}{\,{\rm Tr}}

\newcommand{\e}{\epsilon}

\renewcommand{\d}{\delta}

\newcommand{\la}{\lambda}
\renewcommand{\a}{\alpha}

\newcommand{\g}{\gamma}

\newcommand{\s}{\sigma}
\newcommand{\p}{\partial}

\newcommand\R{{\mathbb{R}}}

\newcommand\C{{\mathcal{C}}}

\renewcommand\H{{\mathbb{H}}}
\newcommand\Hh{{\mathcal{H}}}
\newcommand\E{{\mathbb{E}}}

\newtheorem{theorem}{Theorem}[section]
\newtheorem{proposition}[theorem]{Proposition}

\theoremstyle{definition}

\theoremstyle{remark}
\newtheorem{remark}[theorem]{Remark}
\numberwithin{equation}{section}

\begin{document}

\title[A Mean Field Game of Portfolio Trading]
{A Mean Field Game of Portfolio Trading And Its Consequences On Perceived Correlations}

\author{Charles-Albert Lehalle}
\address{C.-A. Lehalle: Capital Fund Management and Imperial College London.}
\email{c.lehalle{@}imperial.ac.uk}
 
\author{Charafeddine Mouzouni}
\address{C. Mouzouni: Univ Lyon, \'Ecole centrale de Lyon, CNRS UMR 5208, Institut Camille Jordan, 36 avenue Guy de Collonge, F-69134 Ecully Cedex, France.}
\thanks{ \emph{Acknowledgement.} C. Mouzouni was supported by LABEX MI- LYON (ANR-10-LABX-0070) of Universit\'e de Lyon, within the program ``Investissements d'Avenir" (ANR-11-IDEX-0007) operated by the French National Research Agency (ANR), and partially supported by project (ANR-16-CE40-0015-01) on Mean Field Games.
Authors would like to thank Pierre Cardaliaguet for a careful reading of large parts of this paper, and Jean-Philippe Bouchaud to have insisted on the fact that results of \cite{cardaliaguet2016mfgcontrols} should be observable using intraday data.}
\email{mouzouni{@}math.univ-lyon1.fr}

\dedicatory{Version: \today}


\begin{abstract}
  This paper goes beyond the optimal trading Mean Field Game model introduced by Pierre Cardaliaguet and Charles-Albert Lehalle in \cite{cardaliaguet2016mfgcontrols}. It starts by extending it to portfolios of correlated instruments. This leads to several original contributions: first that hedging strategies naturally stem from optimal liquidation schemes on portfolios. Second we show the influence of trading flows on naive estimates of intraday volatility and correlations. Focussing on this important relation, we exhibit a closed form formula expressing standard estimates of correlations as a function of the underlying correlations and the initial imbalance of large orders, via the optimal flows of our mean field game between traders. To support our theoretical findings, we use a real dataset of 176 US stocks from January to December 2014 sampled every 5 minutes to {analyze the influence of the daily flows on the observed correlations. Finally, we propose a toy model based approach to calibrate our MFG model on data.} 
  
%
  \bigskip
  
 
\end{abstract}


\maketitle

\section{Introduction}

Optimal liquidation emerged as an academic field with two seminal papers:
one \cite{almgren1999value} focussed on the balance between trading fast (to minimize the uncertainty of the obtained price) and trading slow (to minimize the ``\emph{market impact}'', i.e. the detrimental influence of the trading pressure on price moves) for one representative instrument; while the other \cite{BLA98} focussed on a portfolio of tradable instruments, shedding light on the interplay with correlations of price returns and market impact.
The last twenty years have seen a lot of proposals to sophisticate the single instrument case (see these reference books \cite{book-Lehalle,book-Jaimungal,Gueant} for typical models and references) but very few on extending it to portfolios of multiple assets (with the notable exception of \cite{gueant2015convex}).
Moreover,  the usual framework for optimal execution is the one of one large privileged agent facing a ``mean-field" or a ``background noise" made of the sum of behaviours of other market participants, and academic literature seldom tackles the strategic interaction of many market participants seeking to execute large orders.

More recently, game theory has been introduced in this field. First around cases with few agents, like in \cite{schied2017high}, and then by \cite{cardaliaguet2016mfgcontrols,citeulike:13586166,caine16trade} relying on Mean Field Games (MFG) to get rid of the combinatorial complexity of games with few players, considering a lot of agents, such that their aggregated behaviour reduces to a ``\emph{anonymous mean field of liquidity}'', shared by all of them.

In this paper, we clearly start within the framework and results obtained by \cite{cardaliaguet2016mfgcontrols} and extend them to the case of a portfolio of tradable instruments. Our agents are the same as in this paper: optimal traders seeking to buy and sell positions given at the start of the day. That for, they rely on the stochastic control problem well defined for one instrument in \cite{book-Jaimungal}, which result turns to be deterministic because of its linear-quadratic nature: minimize the cost of the trading under risk-averse conditions and a terminal cost.
This framework can be compared to the one used by \cite{BLA98} in their section on portfolio, with a diagonal matrix for the market impact, and in a game played by a continuum of agents.
Note that in all these papers, including ours, the time scale is large enough to not take into account orderbook dynamics, and small enough to be used by traders and dealing desks; our typical terminal time goes from one hour to several days, and time steps have to be read in minutes.
In their paper, Cardaliaguet and Lehalle have shown how a continuum of such agents with heterogenous preferences can emulate a mix of typical brokers (having a large risk aversion and terminal cost), and opportunistic traders (with a low risk aversion). It will be the same for us.
But while their paper only addresses the strategic behavior of investors on a one single financial instrument this one handles the case of a portfolio of correlated assets. In the real applications, a financial instrument is rarely traded on its own; most investors construct diversified or hedged portfolios or index trackers by simultaneously buying and selling a large number of assets.



This has motivated the present work in which we introduce an extension of the initial Cardaliaguet-Lehalle framework to the case of a multi-asset portfolio. On the one hand, this extension allows to cover a new type of trading strategies, such as Program Trading (executing large baskets of stocks), Arbitrage Strategies (which aims to benefit from discrepancies in the dynamics of two or more assets), Hedging Strategies (where a round trip on a second asset -- typically a very liquid one -- can be used to partially hedge the price risk in the execution process of a given asset), and Index Tracking (i.e. following the composition given by a formula, like in factor investing, or simply following the market capitalization of a list of instruments). On the other hand, it enables us to understand the dependence structure between the market orders flows at ``equilibrium", and assess their influence on standard estimates of the covariance (or correlation) matrix of asset returns. These questions were independently raised by some authors and studied in seldom empirical and theoretical works (see e.g. \cite{cont2013running,benzaquen2017dissecting, hasbrouck2001common, boulatov2012informed, mastromatteo2017trading} and references therein).

Following the seminal paper \cite{cardaliaguet2016mfgcontrols}, we assume that the market impact is either instantaneous or permanent, and that the public prices -- of all assets -- are influenced by the permanent market impact of all market participants. Conversely, since the agents are affected by the public prices, they aim to anticipate the ``market mean field'' (i.e. the market trend due to the market impact of the mean field of all agents)  
by using all the information they have in order to minimize their exposition to the other agents' impact. As explained in \cite{cardaliaguet2016mfgcontrols} this leads to a Nash equilibrium configuration of MFG type, in which all agents anticipate the average trading speed of the population and adjust their execution accordingly. We refer the reader to Section \ref{first-section} for a more detailed explanation of the Mean Field Game model. In the context of a MFG with multi-asset portfolio, the strategic interaction between the agents during a trading day leads to a non-trivial relationship between the assets' order flows, which in turn generates a non-trivial impact on the intraday covariance (or correlation) matrix of asset returns. 
In Section \ref{cov-matrix-secction}, we provide an exact formula for the excess covariance matrix of returns that is endogenously generated by the trading activity, and we show that the magnitude of this effect is more significant when the market impact is large. This means for an highly crowded market, illiquid products or large initial orders 
(cf. Section \ref{cov-matrix-secction}).
These results can be related to the ones of \cite{cont2013running}, except that in this paper we do not focus our attention on distressed sells only; we are able to capture the influence of the usual variations of trading flows to deformations of the naive estimate of the covariance matrix of a portfolio of assets that are simultaneously traded. {We also carry out several numerical simulations and apply our results in an empirical analysis which is conducted on a database of market data from January to December 2014 for a pool of 176 US stocks.} At first, we exhibit the theoretical  relation between the intraday covariance matrix of net traded flows and the standard intraday covariance matrix is increasing, then we use this relation to estimate some parameters of our model, including the market impact coefficients (cf. Section \ref{cov-matrix-secction}). Next, we normalize the covariance matrix of returns to compute the  intraday median diagonal pattern (across diagonal terms), and the intraday median off-diagonal pattern (across off-diagonal terms) (cf. Section \ref{section33-Empirical}), as a way of characterizing the typical intraday evolution for diagonal and off-diagonal terms. It allows us to obtain empirically the well-known intraday pattern of volatility that is in line with our model, and we show that it flattens out as the typical size of transactions diminishes. In such a case the empirical volatility is close to its ``fundamental"  value (cf. Figure \ref{Figure4-Intrinsec}). 
{Finally, we propose a toy model based approach to calibrate our MFG model on data.}

This paper is structured as follows: in Section \ref{first-section} we formulate the problem of optimal execution of a multi-asset portfolio inside a Mean Field Game. We derive a MFG system of PDEs and prove uniqueness of solutions to that system for a general Hamiltonian function. Then we construct a regular solution in the quadratic framework, which will be considered throughout the rest of the paper. Next, we provide a convenient numerical scheme to compute the solution of the MFG system, and present several examples of an agent's optimal trading path, and the average trading path of the population. Section \ref{cov-matrix-secction} is devoted to the analysis of the crowd's trading impact on the intraday covariance matrix of returns. At the MFG equilibrium configuration, we derive a formula for the impact of assets' order flows on the dependence structure of asset returns. Next, we carry out numerical simulations to illustrate this fact, and apply our results in
an empirical analysis on a pool of $176$ US stocks. 


\newpage
\section{Optimal Portfolio Trading  Within The Crowd}\label{first-section}

\subsection{The Mean Field Game Model}\label{The Mean Field Game Model}
Consider a continuum of investors (agents), which are indexed by a parameter $a$.
Each agent has to trade a portfolio corresponding to instructions given by a portfolio manager. Think about a continuum of brokers or dealing desks executing large orders given by their clients. The portfolios are made of desired positions in a universe of $d$ different
stocks (or any financial assets). 
The initial position of any agent $a$ is denoted by ${\bf q}_0^a:=(q_0^{1,a},...,q_0^{d,a})$. For any $i$, when the initial inventory $q_0^{i,a}$ is positive,  it means the agent has to sell this number of shares (or contracts) whereas when it is negative, the agent has to sell this amount.
Given a common horizon $T>0$, we suppose that all the investors have to sell or buy within the trading period $[0,T]$. This means the agent has to sell this number of shares (or contracts) whereas when it is negative, the agent has to buy this amount.

The intraday position of each investor $a$ is modeled by a $\R^d$-valued process $({\bf q}_t^a)_{t \in[0,T]}$ which has the following dynamics:
$$
\dd {\bf q}_t^{a}={\bf v}_t^{a} \dd t,\quad {\bf q}^a(0)={\bf q}_0^a.
$$
The investor controls its trading speed $({\bf v}_t^{a})_t:= (v_t^{1,a},...,v_t^{d,a})_t$ through time, in order to achieve its trading goal.
Following the standard optimal liquidation literature, 
we assume that, for each stock,  the dynamics of the mid-price can be written as:
\begin{equation}\label{price-dynamics}
\dd S_t^{i} = \s_i \dd W_t^i + \a_i \mu_t^{i} \dd t, \quad i=1,...,d;
\end{equation}
where $\s_i>0$  is the arithmetic volatility of the $i^{th}$ stock, and  $\a_1,...,\a_d$ are nonnegative scalars modeling the magnitude of the permanent market impact. Here $(W_t^1,...,W_t^d)_{t\geq 0}$ are $d$ correlated Wiener processes, and the process 
$(\boldsymbol{\mu}_t)_{t\in[0,T]}:= (\mu_t^1,...,\mu_t^d)_t$ corresponds to the average trading speed of all investors across the portfolio of assets.
Throughout, we shall denote by $\Sigma$ the covariance matrix of the $d$-dimensional process $({\bf W}_t)_{t\in [0,T]} := (\s_1 W_t^1,...,\s_d W_t^d)_{t \in [0,T]}$ and suppose that \emph{$\Sigma$ is not singular.}
      
The performance of any investor $a$ is related to the amount of cash generated throughout the trading process.
Given the price vector $({\bf S}_t)_{t\in[0,T]} := (S_t^1,...,S_t^d)_{t\in[0,T]}$,  
the amount of cash $(X_t^a)_{t\in[0,T]}$ on the account of the trader $a$ is given by: 
$$
 X_t^a =  -  \int_0^t {\bf v}_s^{a} \cdot {\bf S}_s \dd s - \sum_{i=1}^d  \int_0^t V_{i} L_{i}\left(\frac{v_s^{i,a}}{V_i} \right) \dd s,
$$
where the positive scalars $V_1, ..., V_d$ denotes the magnitude of daily market liquidity (in practice the average volume traded each day can be used as a proxy for this parameter) of each asset. Here
$L_1,...,L_d$ are the execution cost functions (similar to the ones of \cite{gueant2015convex}), modelling the instantaneous component of market impact, which takes part in 
the average cost of trading. The family of functions  $L_i:\R \to \R$ are assumed to fulfil the following set of assumptions:
\begin{itemize}
\item  $L_i(0)=0$;
\item $L_i$ is strictly convex and nonnegative;
\item $L_i$ is asymptotically super-linear, i.e. $\lim_{|p|\to + \infty} \frac{L_{i}(p)}{|p|} = +\infty.$
\end{itemize}

The initial Cardaliaguet-Lehalle model \cite{cardaliaguet2016mfgcontrols}, corresponds to $d=1$, and a quadratic liquidity function of the form $L(p)=\kappa |p|^{2}$. 

In this paper, we consider a reward function 
that is similar to \cite{cardaliaguet2016mfgcontrols}, and corresponding to \emph {Implementation Shortfall (IS) orders}.  In this specific case the reward function of any investor $a$  is given by:
\begin{eqnarray}\label{equation-IS}
\\ \nonumber
U^ a(t,x,{\bf s},{\bf q}; \boldsymbol{\mu}) := \sup_{\bf v}\E_{x,{\bf s},{\bf q}}\left( X_T^a +  {\bf q}_T^{a} \cdot ({\bf S}_T - {\bf A}^a {\bf q}_T^{a}) - \frac{\g^a}{2} \int_t^T {\bf q}_s^a \cdot  \Sigma  {\bf q}_s^a  \dd s   \right),
\end{eqnarray}
where ${\bf A}^a := {\rm \bf diag}(A_1^a,...,A_d^a)$, $A_i^a >0$, and $\g^a$ is a non-negative scalar which quantifies the investor's risk aversion. 
That is when $\g^a=0$ the investor is indifferent about holding inventories through time, while when $\g^a$ is large the investor attempt to liquidate as fast as possible. The quadratic term ${\bf q}_T^{a}\cdot({\bf S}_T - {\bf A}^a {\bf q}_T^{a})$ penalizes  non-zero terminal inventories. One should note that the expression of the reward function \eqref{equation-IS} is derived by considering that agents are risk-averse with CARA utility function. We omit the details and refer the reader to \cite[Chapter 5]{Gueant}.

The Hamilton-Jacobi equation associated to \eqref{equation-IS} is
\begin{eqnarray} \nonumber
&&0= \p_t U^ a - \frac{\g^a}{2}  {\bf q} \cdot \Sigma {\bf q} +\frac{1}{2}\Tr\left( \Sigma D_{\bf s}^2 U^ a  \right)+ \mathbb A\boldsymbol{\mu} \cdot \nabla_{\bf s} U^ a     \\ \nonumber
&& \quad  \quad\quad\quad\quad \quad\quad\quad+\sup_{\bf v}\left\{  {\bf v} \cdot \nabla_{\bf q} U^ a -  \left( {\bf v}\cdot {\bf s}+ \sum_{i=1}^{d} V_i L_i \left( \frac{v^i}{V_i}  \right)  \right) \nabla_x U^ a  \right\},
\end{eqnarray}
with the terminal condition
$$
U^ a(T,x,{\bf s},{\bf q}; \boldsymbol{\mu})= x + {\bf q}\cdot({\bf s}-{\bf A}^a{\bf q}).
$$
In all this paper we set $\mathbb A := {\rm \bf diag}(\a_1,...,\a_d)$. Due to the simplifications that we will obtain afterwards, we suppose that $\boldsymbol{\mu}=(\boldsymbol{\mu}_{t})_{t\in [0,T]}$ is a deterministic process, so that the HJB equation above is deterministic. When $\boldsymbol{\mu}$ is a random process, that is adapted to the natural filtration of $({\bf W}_t)_{t\in [0,T]}$, we obtain a stochastic backward HJB equation which requires a specific treatment (cf. \cite{master-equation}).

Following the approach of \cite{cardaliaguet2016mfgcontrols}, we consider the following ersatz:
$$
U^a(t,x,{\bf s},{\bf q}; \boldsymbol{\mu}) = x+{\bf q}\cdot{\bf s} + u^a(t,{\bf q}; \boldsymbol{\mu}),
$$
which entails the following HJB equation for $u^a$:
\begin{equation}\label{HJB-reduced}
\frac{\g^a}{2} {\bf q} \cdot \Sigma {\bf q}=\p_t u^a  + \mathbb A \boldsymbol{\mu} \cdot {\bf q}  +\sup_{\bf v}\left\{  {\bf v} \cdot \nabla_q u^ a -  \sum_{i=1}^{d} V_i L_i \left( \frac{v^i}{V_i}   \right)  \right\},
\end{equation}
endowed with the terminal condition:
$$
u_T^a = - {\bf A}^a{\bf q}\cdot {\bf q}.
$$

For any $i=1,...,d$, let $H_i$ be the Legendre-Fenchel transform of the function $L_i$ that is given by
$$
H_i (p) := \sup_\rho p\rho - L_i(\rho).
$$
Since the maps $(L_i)_{1\leq i \leq d}$ are strictly convex, $(H_i)_{1\leq i \leq d}$ are functions of class $\C^1$, and the optimal feedback strategies associated to  \eqref{HJB-reduced} are given by
$$
v^{i,a}(t,{\bf q}) := V_i \dot H_{i} \left(\p_{q_i} u^a (t,{\bf q})\right),
$$
where $\dot H_i$ denotes the first derivative of $H_i$.
Therefore, the Mean Field Game system associated to the above problem reads: 
\begin{equation}
\label{MFG-system-IS} 
\left\{
\begin{aligned}
& \frac{\g^a}{2} {\bf q} \cdot \Sigma {\bf q}=\p_t u^a  + \mathbb A\boldsymbol{\mu} \cdot {\bf q} + \sum_{i=1}^{d} V_i H_i \left(\p_{q_i} u^a (t,{\bf q})\right) \\ 
\\
& \p_t m +\sum_{i=1}^{d}V_i \p_{q_i}\left( m  \dot H_{i} \left(\p_{q_i} u^a (t, {\bf q})\right) \right) =0 \\
\\
& \mu_t^i = \int_{(q,a)} V_i \dot H_{i} \left(\p_{q_i} u^a (t,{\bf q})\right)  m(t,\dd {\bf q}, \dd a) \\
\\
& m(0,\dd {\bf q}, \dd a)= m_0(\dd {\bf q}, \dd a), \quad u_T^a = -{\bf A}^a{\bf q}\cdot {\bf q}.
\end{aligned}
\right.
\end{equation} 
The Mean Field Game system \eqref{MFG-system-IS} describes a Nash equilibrium configuration, with infinitely many well-informed market investors: any individual player anticipates the right average trading flow on the trading period $[0,T]$, and computes his optimal strategy accordingly. Observe that we make a strong assumption by supposing that the considered group of investors has a precise knowledge of market mean field. In reality this knowledge is only partial or approximate.

Well-posedness for system \eqref{MFG-system-IS} is investigated in \cite{cardaliaguet2016mfgcontrols} within the general framework of \emph{Mean Field Games of Controls}. In this work, we provide simpler arguments to deal with the specific cases of our study. We shall suppose that $(H_i)_{1\leq i \leq d}$ are of class $\C^2$ and satisfy the following condition: 
\begin{equation}\label{assumption-on-H_i}
 \forall i=1,...,d, \ \forall p \in \R, \quad  C_{0}^{-1} \leq  \ddot H_i (p) \leq C_{0},
\end{equation}
for some $C_0>0$, and $m_0$ is a probability density \emph{with a finite second order moment}. Moreover, we suppose that the investors' index varies in a closed subset $D\subset \R$.

We say that  $(u^a,m)_{a\in D}$ is a solution to the MFG system \eqref{MFG-system-IS} if the following hold:
\begin{itemize}
\item $u^a \in \C^{1,2}([0,T]\times \R)$, for a.e $a\in D$, and $m$ in $\C([0,T] ; L^1(\R \times D))$;
\item the equation for $u^a$ holds in the classical sens, while the equation for $m$ holds in the sense of distribution;
\item for any $t\in [0,T]$,
\begin{equation}\label{condition-solution-moment-bound}
\int_{\R\times D} |{\bf q}| \dd m(t, \dd {\bf q}, \dd a) < \infty, \ \mbox{ and } \ \left| \nabla_q  u^a(t,{\bf q}) \right| \leq C_1(1+|{\bf q}|),
\end{equation}
for some $C_1>0$.
\end{itemize}

Let us start with the following remark on the uniqueness of solutions to \eqref{MFG-system-IS}.

\begin{proposition}\label{uniqueness-MFG-system-prop11}
Under the above assumptions, system \eqref{MFG-system-IS} has at most one solution.
\end{proposition}
\begin{proof}
Let $(u_1^a,m_1)_{a\in D}$ and $(u_2^a , m_2)_{a \in D}$ be two solutions to \eqref{MFG-system-IS}, and set $\bar u^a := u_1^a - u_2^a$, $\bar m := m_1-m_2$. At first, let us assume that $m_1$, $m_2$ are smooth so that the computations below holds. By using system \eqref{MFG-system-IS}, we have:

\begin{eqnarray}\label{energy-identity} \\ \nonumber
\frac{\dd}{\dd t} \int_{(q,a)} \bar u^a \bar m &=& -\int_{(q,a)} \bar{m}\left\{ \sum_{i=1}^{d}  V_i \left( H_i \left(\p_{q_i} u_1\right)-H_i \left(\p_{q_i} u_2\right)  \right) +\mathbb A (\boldsymbol{\mu}_1-\boldsymbol{\mu}_2) \cdot \bf q   \right\} \\ \nonumber
&& \quad \quad - \int_{(q,a)} \bar{u}\left\{ \sum_{i=1}^{d}V_i \left( \p_{q_i}\left( m_1  \dot H_{i} \left(\p_{q_i} u_1\right) \right) - \p_{q_i}\left( m_2  \dot H_{i} \left(\p_{q_i} u_2\right) \right)\right) \right\},
\end{eqnarray}
where $\boldsymbol{\mu}_1$, $\boldsymbol{\mu}_2$ correspond respectively to $(u_1^a,m_1)_{a\in D}$ and $(u_2^a , m_2)_{a \in D}$.

On the one hand, note that
$$
\int_{(q,a)} \bar{m}  \mathbb A (\boldsymbol{\mu}_1-\boldsymbol{\mu}_2) \cdot {\bf q}   = \frac{1}{2}\frac{\dd}{\dd t} \mathbb A \bar {\bf E} \cdot  \bar {\bf E},
\quad
\mbox{ where }
\quad
\bar {\bf E}(t) :=\int_{(q,a)} {\bf q} \dd \bar{m}(t).
$$
This follows from 
$$
\frac{\dd}{\dd t} \bar {\bf E} = \boldsymbol{\mu}_1-\boldsymbol{\mu}_2,
$$
which is in turn  obtained from system \eqref{MFG-system-IS} after an integration by parts.

On the other hand, by virtue of \eqref{assumption-on-H_i} we have
\begin{eqnarray}\nonumber
&&\sum_{i=1}^{d} V_i  \int \left( \bar{m}    \left( H_i \left(\p_{q_i} u_1\right)-H_i \left(\p_{q_i} u_2\right)  \right) 
- \p_{q_i} \bar{u} \left( m_1 \dot H_{i} \left(\p_{q_i} u_1\right) - m_2  \dot H_{i} \left(\p_{q_i} u_2\right) \right)
 \right) \\ \nonumber
 && \quad = -\sum_{i=1}^{d}V_i \int \left(  m_1   \left( H_i \left(\p_{q_i} u_2\right)-H_i \left(\p_{q_i} u_1\right)-  \dot H_{i} \left(\p_{q_i} u_1\right)  \p_{q_i} (u_2-u_1)  \right)   \right) \\ \nonumber
 &&\quad -  \sum_{i=1}^{d} V_i \int  \left(  m_2   \left( H_i \left(\p_{q_i} u_1\right)-H_i \left(\p_{q_i} u_2\right)- \dot H_{i} \left(\p_{q_i} u_2\right)  \p_{q_i} (u_1-u_2) \right)   \right) \\ \nonumber
 && \quad \leq - \min_{1\leq i \leq d} V_i \int_{(q,a)} \frac{(m_1+m_2)}{2C_0}\left| \nabla_{\bf q} u_1-\nabla_{\bf q} u_2\right|^2.
\end{eqnarray}
Therefore, \eqref{energy-identity} provides
\begin{equation}\label{energy-identity-uniqueness-proof}
\min_{1\leq i \leq d} V_i \int_0^T  \int_{(q,a)} \left| \nabla_{\bf q} u_1(s)-\nabla_{\bf q} u_2(s)\right|^2 \dd (m_1+m_2)\dd s + \frac{C}{2}\mathbb A\bar {\bf E}(T)\cdot \bar {\bf E}(T) =0.
\end{equation}
By using a standard regularization process, identity \eqref{energy-identity-uniqueness-proof} holds true for any solutions $(u_1^a,m_1)_{a\in D}$ and $(u_2^a , m_2)_{a \in D}$ of \eqref{MFG-system-IS}. Thus, one can use this identity to deduce that $\nabla_q u_1\equiv \nabla_q u_2$ on $\left\{ m_1>0 \right\}\cup \left\{ m_2>0 \right\}$, so that $m_1,m_2$ solve the same transport equation:
$$
\p_t \nu +\sum_{i=1}^{d}V_i \p_{q_i}\left( \nu  \dot H_{i} \left(\p_{q_i} u_1^a (t, {\bf q})\right) \right) =0 , \quad \nu_{t=0}=m_0.
$$
This entails $m_1\equiv m_2$  and so $u_1\equiv u_2$, by virtue of our regularity assumptions.
\end{proof}


\subsection{Quadratic Liquidity Functions}\label{Quadratic Liquidity Functions}
In practice the liquidity function is often chosen as strictly convex power function of the form: $L(p)=\eta |p|^{1+\phi} + \omega |p|$, with $\eta, \phi, \omega >0$. The additional term $\omega |p|$ captures proportional costs such as the bid-ask spread, taxes, fees paid to brokers, trading venues and custodians
 \cite{Gueant}. The quadratic case ($\phi=1$) -- that is also considered in \cite{cardaliaguet2016mfgcontrols} -- is particularly interesting because it induces some considerable simplifications and allows to compute the solutions at a relatively low cost. 
Throughout the rest of this paper, we suppose that the liquidity functions take the following simple form:
\begin{equation}\label{definition-quadratic-liquidity-functions}
L^{i}(p)=\eta_i |p|^2 \ \ \mbox{ where } \ \ \eta_i>0, \ \ \ i=1,...,d.
\end{equation}

Following the approach of \cite{cardaliaguet2016mfgcontrols}, we start by setting $ \bar m_0 (\dd a) := \int_{\bf q} m_0(\dd {\bf q}, \dd a)$. We shall suppose that
\begin{equation}\label{desintegrate-assumption}
\bar m_0(a) \neq 0, \quad \mbox{ for a.e } a\in D,
\end{equation}
and that investors do not change their preference parameter $a$ over time. Thus, we always have $\int_q m(t,\dd {\bf q}, \dd a) = \bar m_0(\dd a)$, so that we can disintegrate $m$ into
$$
m(t,\dd {\bf q}, \dd a)=m^a(t, \dd {\bf q}) \bar m_0(\dd a),
$$
where $m^a(t, \dd {\bf q})$ is a probability measure in ${\bf q}$ for $\bar m_0$-almost any $a$. Let us now define the following process which plays an important role in our analysis:
$$
{\bf E}^a(t):= \int_q {\bf q}  \ m^a(t,\dd {\bf q}) \quad \forall t\in [0,T], \ \mbox{ for a.e } a\in D,
$$
and we shall denote by $E^{a,1},...,E^{a,d}$ the components of ${\bf E}^a$. By virtue of the PDE satisfied by $m$, observe that ${\bf E}^a$ satisfies the following:
\begin{eqnarray} \label{E-derivative-mu-comput}
\dot {\bf E}^{a}(t) &=& \int_q {\bf q} \ \p_t m^a(t,\dd {\bf q}) \\ \nonumber
&=& \int_q \left( \frac{V_i}{2\eta_i} \p_{q_i} u^a(t, {\bf q}) \right)_{1\leq i \leq d} m^a(t,\dd {\bf q}),
\end{eqnarray}
so that
\begin{equation}\label{E-derivative-mu}
\boldsymbol{\mu}_t= \int_a \dot {\bf E}^{a}(t) \dd \bar m_0( a).
\end{equation}

Due to the existence of linear and quadratic terms in the equation satisfied by $u^a$, we expect the solution to have the following form:
\begin{equation}\label{ertsaz-small-u}
u^a(t, {\bf q} )= h_a(t) + {\bf q}' \cdot \Hh_a(t) + \frac{1}{2} {\bf q}'\cdot\mathbb H_a(t) \cdot {\bf q}
\end{equation}
where $h_a(t)$ is $\R$-valued function, $\Hh_a(t):=(\Hh_a^i(t))_{1\leq i \leq d}$ is $\R^d$-valued function, and the map $\mathbb H_a(t):=(\mathbb H_a^{i,j}(t))_{1\leq i,j \leq d}$ take values in the set of $\R^{d\times d}$-symmetric matrices. Inserting \eqref{ertsaz-small-u} in the HJB equation of \eqref{MFG-system-IS} and collecting like terms in ${\bf q}$ leads to the following coupled system of BODEs:
\begin{equation}
\label{identification-1} 
\left\{
\begin{aligned}
& \dot h_{a}= -  \mathbb V\Hh_a \cdot \Hh_a \\ 
&  \dot \Hh_{a}=-\mathbb A\boldsymbol{\mu} - 2 \mathbb H_a\mathbb V\Hh_a \\
& \dot {\mathbb H}_{a}= - 2\mathbb H_a \mathbb V \mathbb H_a +\g^a \Sigma \\ 
&h_a(T)=0, \ \ \Hh_a(T)=0, \ \ \mathbb H_a(T)=-2{\bf A}^a,
\end{aligned}
\right.
\end{equation} 
where $\mathbb V:=\mbox{diag}\left(\frac{V_1}{4\eta_1},..., \frac{V_d}{4\eta_d}\right)$. In order to solve completely \eqref{identification-1} we need to know $\boldsymbol{\mu}$, or the process $\dot {\bf E}^{a}$  thanks to   \eqref{E-derivative-mu}. Thus, one needs an additional equation to completely solve the problem. 

By virtue of \eqref{E-derivative-mu-comput}, we have
\begin{equation}\label{expression-speed}
\dot {\bf E}^{a} = 2\mathbb V \Hh_a +2 \mathbb V\mathbb H_a {\bf E}^a.
\end{equation}
By combining this equation with system \eqref{identification-1} one obtains the following FBODE:
\begin{equation}
\label{reduced-equation-IS-general}
\left\{
\begin{aligned}
& \ddot{\bf E}^{a}= -2\mathbb V\mathbb A \int_a \dot {\bf E}^{a} \dd \bar{m}_0(a)   +2\g_a \mathbb V\Sigma{\bf E}^a \\ 
& {\bf E}^a(0)=E_0^a:=\int_q {\bf q}   m_0({\bf q},a)/\bar m_0(a)\\ 
&\dot{\bf E}^{a}(T)+4\mathbb V {\bf A}^a {\bf E}^a(T)=0.
\end{aligned}
\right.
\end{equation} 
This system is a generalized form of the one that is studied in  \cite{cardaliaguet2016mfgcontrols}, and summarizes the whole market mean field. Observe that the permanent market impact acts as a friction term while the market risk terms act as a pushing force toward a faster execution. The investors heterogeneity is taken into account in the first derivative term, which means that the contribution of all the market participants to the average trading flow is already anticipated by all agents. 

System \eqref{reduced-equation-IS-general} is our starting point to solve the MFG system \eqref{MFG-system-IS} in the quadratic case.
Due to the forward-backward structure of system \eqref{reduced-equation-IS-general}, we need a smallness condition on $\mathbb A$ in order to construct a solution. This assumption is also considered in \cite{cardaliaguet2016mfgcontrols}, and is not problematic from a modeling standpoint since $| \mathbb A |$ is generally small in applications (cf. Section \ref{section33-Empirical}). Let us present the construction of solutions to system \eqref{reduced-equation-IS-general}. 

\begin{proposition}\label{proposition-existence-Enonlocal}
Suppose that ${\bf A}^a, \g_a \in L^{\infty}(D)$, then there exists $\a_0>0$ such that, for $|\mathbb A| \leq \a_0$, the following hold: 
\begin{enumerate}[ label=(\roman*)]
\item there exists a unique process ${\bf E}^a$ in $L_{\bar m_0}^1(D; \C^1 ([0,T]))$ which solves system \eqref{reduced-equation-IS-general}; 
\item there exists a constant $C_2>0$, such that
\begin{equation}\label{estimate-trading-order-flow-speed}
\sup_{0\leq w \leq T}| \boldsymbol{\mu}_w | \leq C_2\left( 1 + \int_a | E_0^a | \dd \bar m_0  \right) e^{C_2 T},
\end{equation}
where $(\boldsymbol{\mu}_t)_{t\in[0,T]}$ is given by  \eqref{E-derivative-mu}.
\end{enumerate}
\end{proposition}
\begin{proof}
At first, note that the solution $\mathbb H_{a}$ to the matrix Riccati equation in \eqref{identification-1} exists on $[0,T]$, is unique, depends only on data, and satisfies (see e.g. \cite{kuvcera1973review})
\begin{equation}\label{estimate-riccati-matrix}
-2 {\bf A}^a - T\g^a \Sigma  \leq  \mathbb H_{a} \leq 0,
\end{equation}
where the order in the above inequality should be understood in the sense of positive symmetric matrices.  Moreover, note that $\Sigma \mathbb V$ and $\mathbb V \Sigma$ are both diagonalizable with non-negative eigenvalues. Thus by using the ODE satisfied by $\mathbb H_a$, we know that $\mathbb H_a \mathbb V$  and  $\mathbb V \mathbb H_a$ are both diagonalizable with a constant change of basis matrix. In particular, it holds that
\begin{equation} \label{commutatoin-prop1.2}
\left[ \mathbb H_a(t) \mathbb V , \int_t^w \mathbb H_a(u) \mathbb V \dd u \right]=\left[ \mathbb V  \mathbb H_a(t) , \int_t^w \mathbb V \mathbb  H_a(u) \dd u \right]=0
\end{equation}
for any $0\leq t,w \leq T$, where the symbol $[B, A]$ denotes the Lie Bracket: $[B, A]=BA-AB$.

Given  $\mathbb H_{a}$, we aim to construct $\dot {\bf E} ^a $ in $L_{\bar m_0}^1(D; \C ([0,T]))$ by solving a fixed point relation, and then deduce 
${\bf E} ^a$.  For that purpose, we start by deriving a fixed point relation for $\dot {\bf E} ^a $. 
By virtue of \eqref{commutatoin-prop1.2}, observe that any solution ${\bf E} ^a$ to \eqref{reduced-equation-IS-general} fulfills \eqref{expression-speed}  with (see e.g. \cite{martin1968exponential})
$$
\Hh_a(t)= \int_t^T \exp\left\{\int_t^w  2 \mathbb H_a(s) \mathbb V \dd s\right\}\mathbb A \int_a \dot {\bf E}^{a}(w) \dd \bar{m}_0(a) \dd w,
$$
so that
\begin{multline*}
{\bf E} ^a(t)= \exp\left\{ \int_0^t 2 \mathbb V \mathbb H_a(w) \dd w \right\} E_0^a\\
+ 2 \mathbb V \int_0^t \exp\left\{ \int_\tau^t 2 \mathbb V \mathbb H_a(w) \dd w \right\} \int_\tau^T \exp\left\{ \int_\tau^w 2  \mathbb H_a(s) \mathbb V \dd s \right\} \mathbb A  \int_a \dot {\bf E}^{a}(w) \dd \bar{m}_0(a) \dd w \dd \tau.
\end{multline*}
By combining this relation with \eqref{expression-speed}, we deduce that $\dot {\bf E} ^a$ satisfies the following fixed point relation:
\begin{multline}\label{second-E-prop-wp}
{\bf x} ^a(t) = \Phi_{\mathbb A}({\bf x}^a)(t) :=  2\mathbb V \mathbb H_a(t) \exp\left\{ \int_0^t 2 \mathbb V \mathbb H_a(w) \dd w \right\} E_0^a\\
+ 4 \mathbb V \mathbb H_a(t) \mathbb V \int_0^t \exp\left\{ \int_\tau^t 2 \mathbb V \mathbb H_a(w) \dd w \right\} \int_\tau^T \exp\left\{ \int_\tau^w 2  \mathbb H_a(s) \mathbb V \dd s \right\} \mathbb A  \int_a {\bf x}^{a}(w) \dd \bar{m}_0(a) \dd w \dd \tau\\
+ 2 \mathbb V \int_t^T  \exp\left\{ \int_t^w 2  \mathbb H_a(w) \mathbb V \dd w \right\} \mathbb A \int_a  {\bf x}^{a}(w) \dd \bar{m}_0(a) \dd w.
\end{multline}
Conversely, one checks that if ${\bf x}^a$  is a solution to the fixed point relation \eqref{second-E-prop-wp}, for a.e. $a\in D$, then ${\bf E}^a(t)= E_0^a + \int_0^t {\bf x}^a(s) \dd s$ is a solution to system \eqref{reduced-equation-IS-general}.

To solve the fixed point relation \eqref{second-E-prop-wp}, one just uses Banach fixed point Theorem on $\Phi_{\mathbb A}: \mathbb X \to \mathbb X$, where $\mathbb X := L_{\bar m_0}^1(D; \C ([0,T]))$. It is clear that $\Phi_{\mathbb A}$ is a contraction for $|\mathbb A|$ small enough: indeed, given ${\bf x},{\bf y} \in \mathbb X$, it holds that:
\begin{equation*}
\left| \Phi_{\mathbb A}({\bf x}^a)(t) - \Phi_{\mathbb A}({\bf y}^a)(t)  \right|  \leq C | \mathbb A | \left\| {\bf x} - {\bf y}  \right\|_{\mathbb X} 
\end{equation*}
where $C>0$ depends only on $T, \| \gamma \|_{\infty}, \| {\bf A} \|_{\infty}, |\mathbb V| $ and $|\Sigma|$. Thus, given the solution ${\bf x}^a$ to \eqref{second-E-prop-wp}, the function ${\bf E}^a(t)= E_0^a + \int_0^t {\bf x}^a(s) \dd s$ solves  \eqref{reduced-equation-IS-general}, and belongs to $L_{\bar m_0}^1(D; \C ([0,T]))$ given that $m_0$ have a finite first order moment. Estimate \eqref{estimate-trading-order-flow-speed} ensues from Gr{\"o}nwall's Lemma.
\end{proof}

We are now in position to solve the MFG system \eqref{MFG-system-IS} in the case of quadratic liquidity functions \eqref{definition-quadratic-liquidity-functions}.
\begin{theorem}
Under  \eqref{definition-quadratic-liquidity-functions}, \eqref{desintegrate-assumption}, and assumptions of Proposition \ref{proposition-existence-Enonlocal}, the Mean Field Game system \eqref{MFG-system-IS} has a unique solution.
\end{theorem}
\begin{proof}
Since \eqref{reduced-equation-IS-general} is solvable thanks to Proposition \ref{proposition-existence-Enonlocal}, we can now solve completely system \eqref{identification-1} and deduce $u^a(t,{\bf q}; \boldsymbol{\mu})$ thanks to \eqref{ertsaz-small-u}. In fact, owing to \eqref{commutatoin-prop1.2} we know that (cf. \cite{martin1968exponential}):
\begin{equation*}
\left\{
\begin{aligned}
& \Hh_a(t) = \int_t^T \exp\left\{\int_t^w  2 \mathbb H_a(s) \mathbb V \dd s\right\}\mathbb A \boldsymbol{\mu}_w \dd w\\
& h_{a}(t)= \int_t^T  \mathbb V\Hh_a(w)\cdot \Hh_a(w) \dd w,
\end{aligned}
\right.
\end{equation*} 
so that the function $u^a(t,{\bf q}; \boldsymbol{\mu})$, that is given by \eqref{ertsaz-small-u}, is $\C^{1,2}([0,T]\times \R)$. Furthermore, by virtue of \eqref{estimate-trading-order-flow-speed}-\eqref{estimate-riccati-matrix}, note that
\begin{equation}\label{linear-growth-estimate-Th13}
\left| \nabla_q  u^a(t,{\bf q}) \right| \leq C(1+|{\bf q}|)
\end{equation}
for some constant $C>0$ which depends only on $T$ and data.

Now, as $u^a$ is regular and satisfies \eqref{linear-growth-estimate-Th13}, we know that the transport equation
$$
\p_t m^a +\sum_{i=1}^{d}V_i \p_{q_i}\left( m^a  \p_{q_i} u^a (t, {\bf q}) \right) =0, \quad m^a(0,\dd {\bf q})= m_0(\dd {\bf q}, \dd a)/ \bar m_0(a)
$$
has a unique weak solution $m^a \in \C([0,T] ; L^1(\R))$ for a.e $a\in D$, so that $m :=m^a \bar m_0$ solves, in the weak sense, the following Cauchy problem:
$$
\p_t m +\sum_{i=1}^{d}V_i \p_{q_i}\left( m  \p_{q_i} u^a (t, {\bf q}) \right) =0, \quad m(0,\dd {\bf q}, \dd a)= m_0(\dd {\bf q}, \dd a).
$$
In addition, one easily checks that $m$ belongs to $\C([0,T] ; L^1(\R\times D))$.

By invoking the uniqueness of solutions to \eqref{reduced-equation-IS-general}, we have
$$
{\bf E}^a(t)= \int_q {\bf q} m^a(t,q) \dd {\bf q} \ \ \mbox{ for a.e } a \in D.
$$
Thus through the same computations as in \eqref{E-derivative-mu-comput} we obtain 

$$
\mu_t^i = \int_{(q,a)} V_i \p_{q_i} u^a (t,{\bf q})  m^a(t, {\bf q}) \bar m_0(a) \dd a \dd {\bf q}, \ \ i = 1,...,d,
$$
so that $(u^a,m)_{a\in D}$ solves the MFG system \eqref{MFG-system-IS}.

By virtue of Proposition \ref{uniqueness-MFG-system-prop11}, any constructed solution is unique. So to conclude the proof, it remains to show that:
$$
\int_{\R\times D} |{\bf q}|  m(t, {\bf q}, a) \dd {\bf q} \dd a < \infty.
$$
For that purpose, let us set $\Psi(t):= \int_{\R\times D} |{\bf q}|^2  m(t, {\bf q}, a) \dd {\bf q} \dd a$. After differentiating $\Psi$ and integrating by parts, we obtain the following ODE that is satisfied by $\Psi$:
$$
\Psi(t) = \Psi(0) + 2 \int_0^t \int_a   \Hh_a(w) \cdot {\bf E}^a(w)  \bar m_0(\dd a) \dd w + 2 \int_0^t \int_a   \left(  \mathbb H_a(w) {\bf q} \cdot  {\bf q} \right) m(t, \dd {\bf q}, \dd  a)  \dd w,
$$
so that
$$
|\Psi(t)| \leq |\Psi(0)| + C \left\{ \sup_{0\leq w \leq T}\left\| {\bf E}^a(w) \right\|_{L_{\bar m_0}^1} + \int_0^t |\Psi(w)| \dd w \right\}
$$
holds thanks to \eqref{estimate-trading-order-flow-speed}-\eqref{estimate-riccati-matrix}.
Hence, as $m_0$ has a finite second order moment, we deduce from Gr{\"o}nwall's Lemma that for any $t\in[0,T]$
$$
\int_{\R\times D} |{\bf q}|^2  m(t, {\bf q}, a) \dd {\bf q} \dd a < \infty,
$$
which in turn entails the desired result.
\end{proof}

\subsection{Stylized Facts \& Numerical Simulations} \label{Section-num-SF}

Let us now comment our results and highlight several stylized facts of the system. 
By virtue of \eqref{ertsaz-small-u}, the optimal trading speed ${\bf v}_a^{\ast}$ is given by:
\begin{eqnarray}\label{explicit-expression-optimal-speed}
{\bf v}_a^{\ast}(t,{\bf q}) &=&  2\mathbb V\mathbb H_a(t){\bf q}+2\mathbb V\Hh_a(t) \\ \nonumber
&=& 2\mathbb V\mathbb H_a(t){\bf q}+ 2\mathbb V \int_t^T \exp\left\{\int_t^w  2 \mathbb H_a(s) \mathbb V \dd s\right\}\mathbb A \boldsymbol{\mu}_w \dd w  \\
&=:& {\bf v}_a^{1,\ast}(t, {\bf q})+ {\bf v}_a^{2,\ast}(t; \boldsymbol{\mu}). \nonumber
\end{eqnarray}
The above expression shows that the optimal trading speed is divided into two distinct parts ${\bf v}_a^{1,\ast}, {\bf v}_a^{2,\ast}$. The first part 
${\bf v}_a^{1,\ast}$ corresponds to the classical Almgren-Chriss solution in the case of a complexe portfolio (cf. \cite{Gueant}). The second part 
${\bf v}_a^{2,\ast}$ adjusts the speed based on the anticipated future average trading on the remainder of the trading window $[t,T]$. Since the matrix 
$\H_a$ is negative, note that the strategy gives more weight to the current expected average trading. Moreover, the contribution of the corrective term decreases as we approach the end of the trading horizon. The correction term aims to take advantage of the anticipated market mean field. 

Let us set  
\begin{equation}\label{definition-G}
\mathbb G_a(t,w):= \exp\left\{\int_t^w  2 \mathbb H_a(s) \mathbb V \dd s\right\}\mathbb A.
\end{equation}
Note that the matrix $\mathbb G_a$ is not necessarily symmetric and could have a different structure than $\mathbb H_a$. In view of the market price dynamics, the trading speed expression shows that an action of an individual investor or trader on asset $i$ could have a direct impact on the price of asset $j$, at least when the two assets are fundamentally correlated, i.e. $\Sigma_{i,j}\neq 0$. This phenomenon of \emph{cross impact} is related to the fact that other traders already anticipates the market mean field and aim to take advantage from that information, especially when asset $j$ is more liquid than asset $i$ (or vice versa). Thus, if an investor is trading as the crowd is expecting her to trade, then she is more likely to get a \emph{``cross-impact"} through the action of the other traders. This fact is empirically addressed in \cite{benzaquen2017dissecting, hasbrouck2001common}.

Another expression of the optimal trading speed can also be derived thanks to \eqref{expression-speed}. In fact, we have that:
\begin{equation}\label{individual-optimal-curve-MFG}
{\bf v}_a^{\ast}(t,q)  = \dot {\bf E}^{a} + 2\mathbb V\mathbb H_a(t)({\bf q}-{\bf E}^{a}).
\end{equation}
The above formulation shows that an individual investor should follow the market mean field but with a correction term which depends on the situation of her inventory relative to the population average inventory.

In order to simplify the presentation, we ignore from now on investors heterogeneity and assume that market participants have identical preferences. Under this assumption, system \eqref{reduced-equation-IS-general} simply reads:
\begin{equation}
\label{reduced-equation-IS-identical}
\left\{
\begin{aligned}
& \ddot {\bf E}=  -2\mathbb V\mathbb A \dot{\bf E} +2\g \mathbb V\Sigma{\bf E} \\ 
& {\bf E}(0)={ E}_0, \ \ \dot {\bf E}(T)+4\mathbb V {\bf A} {\bf E}(T)=0.
\end{aligned}
\right.
\end{equation} 
Given a discretization step $\d t=N^{-1}$, the solution of \eqref{reduced-equation-IS-identical} is approached by a sequence 
$\left(x_k, y_k\right)_{0\leq k \leq N}$ according to the following implicit scheme:
\begin{equation}
\nonumber
\left\{
\begin{aligned}
& x_0=E_0 \\ 
& x_k - x_{k-1}-\d t  y_{k-1}=0, \ \ k=1,...,N \\
& y_k - y_{k-1}-\d t\left( 2\g \mathbb V\Sigma x_k -2\mathbb V\mathbb A y_k \right)=0, \ \ k=1,...,N\\
& 4 \mathbb V{\bf A}x_{N}+y_N=0.
\end{aligned}
\right.
\end{equation} 
Hence, computing an approximate solution to system \eqref{reduced-equation-IS-identical} reduces to solving a straightforward linear system.
One checks that under conditions of Proposition \ref{proposition-existence-Enonlocal}, the above numerical scheme converges and is stable.

Now, we can present some examples by using the above numerical method. We consider a portfolio containing three assets (Asset 1, Asset 2, Asset 3) with the following characteristics: 
\begin{itemize}
\item  $\sigma_1=\sigma_3=0.3 \   {\$.day}^{-1/2}.{share}^{-1}$, $\sigma_2=1 \  { \$.day}^{-1/2}.{share}^{-1}$;
\item $V_1=2,000,000 \ { share.day}^{-1}$, $V_2=V_3=5,000,000 \ { share.day}^{-1}$;
\item $\eta_1=\eta_2=0.1\  { \$.share}^{-1}$, $\eta_3=0.4\  { \$.share}^{-1}$, 
$A_1=A_2=2.5 \ {\$.day}^{-1}.{share}^{-1}$;
\item $\a_1=\a_2=8 \times 10^{-4} \ {\$.share}^{-1}$, $\a_3= 6 \times 10^{-4} \ {\$.share}^{-1}$.
\end{itemize}

In Figure \ref{figure1-liquidation}-\ref{figure2-individual-player-1}, we consider a market with the initial average inventories $E_0^1=100,000$, $E_0^2=50,000$, and $E_0^3= -25,000$ shares, for Asset 1, Asset 2,  and Asset 3  respectively. In this example, we suppose that the correlation between the price increments of Asset 1 and Asset 2 is 80 \%, and we set $\g=5 \times 10^{-5} \ \mbox{\$}^{-1}$ except for Figure \ref{figure1-hedging}.

Figure \ref{figure1-liquidation} shows that changing the permanent market impact prefactors $(\alpha_k)_{1\leq k\leq 3}$ has a significant influence on the average execution speed. This fact was pointed out in \cite{cardaliaguet2016mfgcontrols}, and is essentially related to the fact that the higher the permanent market impact parameter the more the anticipated influence of the other market participants become important. Namely, when $\a_k$ is large, traders anticipate a more significant pressure on the price of Asset $k$, and adjust their trading speed. On the other hand, dynamics of Asset 2 shows that  the higher the market liquidity the faster is the execution. This is expected since the more liquid the faster assets are traded. Finally, dynamics of Asset 3 shows that traders accelerate their execution on volatile asset. It corresponds to a natural reaction due to risk aversion; a trader will try to reduce his exposure to the \emph{more risky} (hence volatile) assets in priority.

\begin{figure}[hbtp] 
\begin{center}
        \subfigure[Market mean field with different parameters]{%
            \label{figure1-liquidation}
            \includegraphics[width=0.5\textwidth]{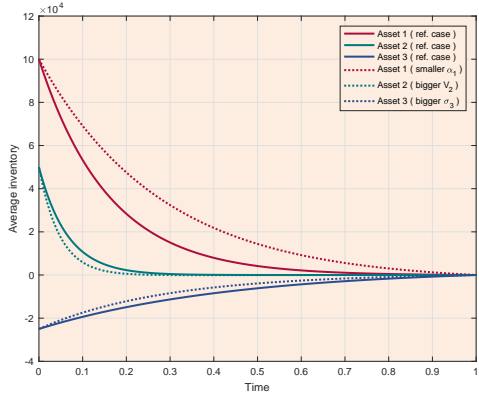}
        }%
        \subfigure[Optimal trading of an individual investor with: $q_0^1=40, 000$,  $q_0^2=0$, and $q_0^3=110, 000$]{%
           \label{figure2-individual-player}
           \includegraphics[width=0.5\textwidth]{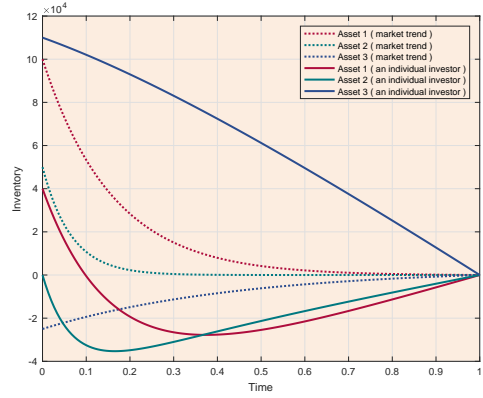} 
        }\\ 
        \subfigure[Market mean field with high risk aversion]{%
            \label{figure1-hedging}
            \includegraphics[width=0.5\textwidth]{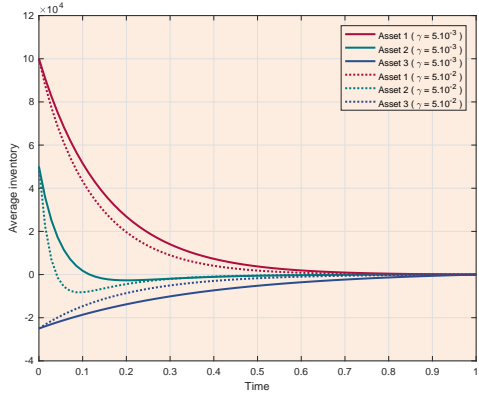}
        }%
        \subfigure[Optimal trading of an individual investor with: $q_0^1=100, 000$, and $q_0^2=q_0^3=0$  ]{%
            \label{figure2-individual-player-1}
            \includegraphics[width=0.5\textwidth]{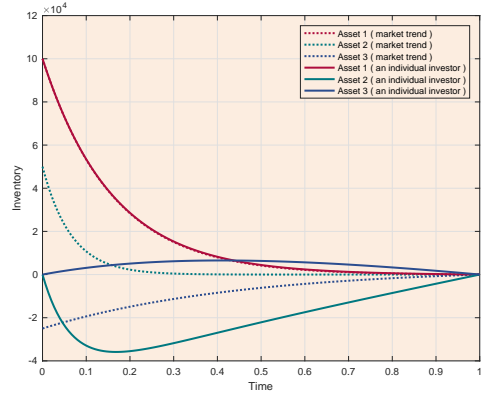}
        }%
        \end{center}
        \caption{%
        Simulated examples of the dynamics of ${\bf E}$ and optimal trading curves of an individual investor. The dashed lines in Figure \ref{figure1-liquidation} correspond to:  $\a_1=6 \times 10^{-4} \ {\$.share}^{-1}$, $V_2=7,000,000 \ { share.day}^{-1}$, and $\sigma_3= 5 \  { \$.day}^{-1/2}.{share}^{-1}$.
     }%
\end{figure} 

Figure \ref{figure1-hedging} illustrates the behavior of the crowd of investors with an increasing risk aversion (higher $\g$). In the two presented scenarios, one can observe that Asset 2 is liquidated very quickly, then a short position is built (around $t=0.05$ for $\g=5 \times10^{-2}\ \mbox{\$}^{-1}$) and it is finally progressively unwound. This exhibits the emergence of a Hedging Strategy: indeed, since Asset 1 and Asset 2 are highly correlated, investors can slow down the execution process for the less liquid asset (Asset 1) to reduce the transaction costs, by using the more liquid asset (Asset 2) to hedge the market risk associated to Asset 1. The trader has an incentive to use such a strategy as soon as the cost of the roundtrip in Asset 2 is smaller than the corresponding reduction of the risk exposure (seen from its reward function $U^ a(t,x,{\bf s},{\bf q}; \boldsymbol{\mu})$ defined by equality (\ref{equation-IS})).

Now, we provide examples of individual players' optimal strategies. We consider two examples: an individual investor with initial inventory 
$q_0^1=40, 000$, $q_0^2=0$, and  $q_0^3=110, 000$ in Figure \ref{figure2-individual-player}; and an individual investor with initial inventory 
$q_0^1=100, 000$, and $q_0^2=q_0^3=0$ in Figure \ref{figure2-individual-player-1}.

In Figure \ref{figure2-individual-player-1} the considered investor starts from $q_0^1=E_0^1$. Hence, by virtue of \eqref{individual-optimal-curve-MFG} her liquidation curve follows exactly the market mean field. Moreover, the investor takes advantage of the anticipated evolution of the market by building favorable positions on Asset 2 and Asset 3: building a short (resp. a long) position on Asset 2 (resp. Asset 3), and buying (resp. selling) back in order to take advantage of price drop (resp. raise) induced by the massive liquidation (resp. purchase). The trading strategies on Asset 2 and Asset 3 are related to the term ${\bf v}^{2,\ast}$ in \eqref{explicit-expression-optimal-speed}. This strategy can be described as a ``Liquidity Arbitrage Strategy''. 

Figure \ref{figure2-individual-player} shows two interesting facts: on the one hand, the individual player builds a short position on Asset 1 after achieving her goal (complete liquidation) in order to take advantage of the market selling pressure; on the other hand, by taking into account the market buying pressure on Asset 1, the investor slows down her liquidation to reduce execution costs since she anticipates no  sustainable price decline.

\section{The Dependence Structure of Asset Returns}\label{cov-matrix-secction}

The main purpose of this section is to analyze the impact of large transactions on the observed covariance matrix between asset returns, by using the Mean Field Game framework of Section \ref{first-section}. For that purpose, we assume a simple model where a continuum of players trade a portfolio of assets on each day, and where the initial distribution of inventories  across the investors $m_0$ changes randomly from one day to another according to some given law of probability. We assume that the price dynamics is given by \eqref{price-dynamics}, and we consider the problem of estimating the covariance matrix of asset returns given a large dataset of  intraday observations of the price. For the sake of simplicity, we ignore investors heterogeneity and assume that market participants have identical preferences. Next, we compare our findings with an empirical analysis on a pool of 176 US stocks sampled every 5 minutes over year 2014 and calibrate our model to market data.

Throughout this section, we denote by $\left<X^2 \right>$ the variance of $X$, and $\left<X,Y \right>$ the
 covariance between $X$ and $Y$, for any two random variables $X,Y$. Moreover, we will call a ``bin'' a slice of 5 minutes. We focused on continuous trading hours because the mechanism of call auctions (i.e. opening and closing auctions is specific). Since US markets open from 9h30 to 16h, our database has 78 bins per day. They will be numbered from 1 to $M$ and indexed by $k$.
 
\subsection{Estimation using Intraday Data} \label{COV-sect1}

We suppose that $E_0$ is a random variable with a given realization on each trading period $[0,T]$, where  $T=1$ day (trading day); and we consider the problem of estimating the covariance matrix of asset returns given the following observations of the price: 
$$
\left\{ \left( {\bf S}_{t_{1,1}}^n,...,{\bf S}_{t_{1,M}}^n  \right) , \left( {\bf S}_{t_{2,1}}^n,...,{\bf S}_{t_{2,M}}^n  \right), ....,  
\left( {\bf S}_{t_{N,1}}^n,...,{\bf S}_{t_{N,M}}^n  \right) \right\}, \quad n=1,...,d
$$
where ${\bf S}_{t_{\ell,k}} ^n$ is the price of asset $n$ in bin $k$ of day $\ell$. We suppose that $t_{\ell,1}=0$, $t_{\ell,M}=T$, for any $1\leq \ell \leq N $, and  $t_{\ell,k}= t_{\ell',k}=t_k$ for any $1\leq k \leq M $,  $1\leq \ell, \ell' \leq N$.

For simplicity, we suppose that the covariance matrix of asset returns between $t_k$ and $t_{k+1}$ is estimated form data by using the following ``naive" estimator :
\begin{equation}\label{covariance-standard-estimation}
C_{[t_k, t_{k+1}]}^{i,j} := \frac{1}{N-1} \sum_{l=1}^N \left( \d S^{i,k,l} -  \overline{\d S}^{i,k}   \right) 
\left( \d S^{j,k,l} -  \overline{\d S}^{j,k}    \right) ,
\end{equation}
where $\d S^{n,k, l} = S_{t_{l, k+1}}^n-S_{t_{l, k}}^n$ and 
$
\overline{\d S}^{n,k}= N^{-1} \sum _{l=1}^N \d S^{n,k, l}
$, $n=i,j$.
We define the correlation matrix as follows:
\begin{equation}\label{correlation-standard-estimation}
R_{[t_k, t_{k+1}]}^{i,j} :=\frac{C_{[t_k, t_{k+1}]}^{i,j}}{\left(C_{[t_k, t_{k+1}]}^{i,i}C_{[t_k, t_{k+1}]}^{j,j} \right)^{1/2}}.
\end{equation}
Suppose that the price dynamics is given by \eqref{price-dynamics}, then the following proposition provides an exact computation of $C_{[t_k, t_{k+1}]}^{i,j}$.
\begin{proposition}\label{proposition5.4-excess-covariance}
Assume that $E_0$ is independent from the process $({\bf W}_t)_{t\in[0,T]}$, then for any $1\leq k \leq M-1$ and $1\leq i,j \leq d$, the following hold:
\begin{equation} \label{asset-returns-covariation}
C_{[t_k, t_{k+1}]}^{i,j} = (t_{k+1}-t_{k})\Sigma_{i,j} + \a_i \a_j \frac{\eta_i \eta_j}{4 V_i V_j}  \Lambda_k^{i,j} + \e_N,
\end{equation}
where  $\e_N \to 0$ as $N\to \infty$, 
\begin{multline*}
\Lambda_k^{i,j} :=  \sum_{1\leq \ell,\ell' \leq d} \left< \theta_k^{i,\ell}, \theta_k^{j,\ell'} \right> +  \sum_{1\leq \ell,\ell' \leq d} \left< \pi_k^{i,\ell}, \theta_k^{j,\ell'} \right> \\ +  \sum_{1\leq \ell,\ell' \leq d} \left< \theta_k^{i,\ell}, \pi_k^{j,\ell'} \right>
+  \sum_{1\leq \ell,\ell' \leq d} \left< \pi_k^{i,\ell}, \pi_k^{j,\ell'} \right>,
\end{multline*}
and
$$
\pi_k^{n,\ell} := \int_{t_k}^{t_{k+1}} \mathbb H^{n,\ell}(s) E^{\ell}(s) \dd s, \quad \theta_k^{n,\ell} := 
\int_{t_k}^{t_{k+1}} \int_{s}^T \mathbb G^{n, \ell}(s,w) \mu^{\ell}(w)\dd w\dd s. 
$$
\end{proposition}
\begin{proof}
Use the exact expression of the price dynamics \eqref{price-dynamics}, the law of large numbers, and the independence between $E_0$ and $({\bf W}_t)_{t\in[0,T]}$ to obtain:
\begin{multline}\label{fundamental-relation-covariance}
C_{[t_k, t_{k+1}]}^{i,j} = \e_N + (t_{k+1}-t_{k})\Sigma_{i,j} \\
+ \a_i \a_j {1\over (N-1)}\sum_{l=1}^{N} \int_{t_k}^{t_{k+1}} \left( \mu_s^{i,l}  - \bar{\mu}_s^{i} \right)  \dd s    \int_{t_k}^{t_{k+1}} \left(  \mu_{s'}^{j,l}  - \bar{\mu}_{s'}^{j} \right) \dd s',  
\end{multline}
where $ \bar{\mu}_u^{n} = N^{-1} \sum_{l=1}^N  \mu_u^{n,l} $, and $\boldsymbol{\mu}^l, {\bf E}^l$ are respectively the realizations of $\boldsymbol{\mu}, {\bf E}$ in day $l$. Now, owing to \eqref{explicit-expression-optimal-speed}-\eqref{definition-G}, we know that
$$
\boldsymbol{\mu}_t^l = 2\mathbb V\mathbb H(t){\bf E}^l(t)
+2\mathbb V\int_t^T\mathbb G(t,w) \boldsymbol{\mu}^l_w \dd w\\
=: \boldsymbol{\nu}^{1,l}(t) + \boldsymbol{\nu}^{2,l}(t).
$$
Thus by setting 
$$ 
\tilde {\boldsymbol{\nu}}_k^{n,l} :=  \int_{t_k}^{t_{k+1}} \left(\boldsymbol{\nu}^{n,l}(s)   -  
N^{-1} \sum_{l=1}^N  \boldsymbol{\nu}^{n,l}(s) \right)   \dd s, \quad n=1,2,
$$
we deduce that
$$
\int_{t_k}^{t_{k+1}} \left( \mu_s^{i,\ell}  - \bar{\mu}_s^{i} \right)  \dd s    \int_{t_k}^{t_{k+1}} \left(  \mu_{s'}^{j,\ell}  - \bar{\mu}_{s'}^{j} \right) \dd s'
= \left( \tilde{\nu}_k^{1,l,i} +  \tilde{\nu}_k^{2,l,i} \right) \left( \tilde{\nu}_k^{1,l,j} + \tilde{\nu}_k^{2,l,j}  \right).
$$
The desired result ensues by noting the existence of estimation noises $\e^1_N, \e^2_N, \e^3_N$ and $\e^4_N$, such that:
\begin{eqnarray*}
(N-1)^{-1}\sum_{l=1}^{N} \tilde{\nu}_k^{1,l,i}\tilde{\nu}_k^{1,l,j} &=& \frac{\eta_i \eta_j}{4 V_i V_j} \sum_{1\leq \ell,\ell' \leq d} \left< \pi_k^{i,\ell}, \pi_k^{j,\ell'} \right> + \e^1_N;\\
(N-1)^{-1}\sum_{l=1}^{N} \tilde{\nu}_k^{2,l,i}\tilde{\nu}_k^{2,l,j} &=& \frac{\eta_i \eta_j}{4 V_i V_j} \sum_{1\leq \ell,\ell' \leq d} \left< \theta_k^{i,\ell}, \theta_k^{j,\ell'} \right> + \e^2_N;\\
(N-1)^{-1}\sum_{l=1}^{N} \tilde{\nu}_k^{1,l,i}\tilde{\nu}_k^{2,l,j} &=& 
\frac{\eta_i \eta_j}{4 V_i V_j} \sum_{1\leq \ell,\ell' \leq d} \left< \pi_k^{i,\ell}, \theta_k^{j,\ell'} \right>+\e^3_N;\\
(N-1)^{-1}\sum_{l=1}^{N} \tilde{\nu}_k^{1,l,j}\tilde{\nu}_k^{2,l,i} &=& \frac{\eta_i \eta_j}{4 V_i V_j} \sum_{1\leq \ell,\ell' \leq d} \left< \theta_k^{i,\ell}, \pi_k^{j,\ell'} \right> 
+ \e^4_N.
\end{eqnarray*}
The proof is complete.
\end{proof}

\begin{remark}
One can easily derive an analogous result  for $\left(C_{[0,T]}^{i,j} \right)_{1\leq i,j \leq d}$. Namely, it holds that:
\begin{equation} \label{asset-returns-covariation-[0,T]}
C_{[0,T]}^{i,j} = T\Sigma_{i,j} + \a_i \a_j\frac{\eta_i \eta_j}{4 V_i V_j}  \Lambda^{i,j} + \e_N,
\end{equation}
where  $\e_N \to 0$ as $N\to \infty$, 
\begin{multline*}
\Lambda^{i,j} :=  \sum_{1\leq \ell,\ell' \leq d} \left< \theta^{i,\ell}, \theta^{j,\ell'} \right> +  \sum_{1\leq \ell,\ell' \leq d} \left< \pi^{i,\ell}, \theta^{j,\ell'} \right>\\ +  \sum_{1\leq \ell,\ell' \leq d} \left< \theta^{i,\ell}, \pi^{j,\ell'} \right>,
+  \sum_{1\leq \ell,\ell' \leq d} \left< \pi^{i,\ell}, \pi^{j,\ell'} \right>,
\end{multline*}
and
$$
\pi^{n,\ell} := \int_{0}^{T} \mathbb H^{n,\ell}(s) E^{\ell}(s) \dd s, \quad \theta^{n,\ell} := 
\int_{0}^{T} \int_{s}^T \mathbb G^{n, \ell}(s,w) \mu^{\ell}(w)\dd w\dd s. 
 $$
\end{remark}

Identities \eqref{asset-returns-covariation} and \eqref{asset-returns-covariation-[0,T]} show that the realized covariance matrix is the sum of the fundamental covariance and an \emph{excess realized covariance matrix} generated by the impact of the crowd of institutional investors' trading strategies. Note on the one hand that the diagonal terms $C^{i,i}$ are always deviated from fundamentals because of the contribution of $\left<(\pi^{i,i})^2\right>$ and $\left<(\theta^{i,i})^2\right>$. On the other hand, since $\mathbb H$ and $\mathbb G$ inherit a structure similar to $\Sigma$, the excess of realized covariance in the off-diagonal terms is non-zero as soon as one -- or both -- of the conditions below is satisfied:
\begin{itemize}
\item there exists $i_0 \neq j_0$ such that $\Sigma_{i_0,j_0} \neq 0$;
\item there exists $i_0 \neq j_0$ such that $\left<E_0^{i_0}, E_0^{j_0} \right>\neq 0$.
\end{itemize}

Moreover, \eqref{asset-returns-covariation} and \eqref{asset-returns-covariation-[0,T]} show that the excess realized covariance can deviate significantly from fundamentals when: \emph{the market impact is large (high crowdedness), the considered assets are highly liquid (small $\eta_i/V_i$), the risk aversion coefficient $\g$ is high, or when the standard deviation of $E_0$ is large}. In addition, since the contribution of $\theta_k^{n,\ell}$ and $\pi_k^{n,\ell}$ vanishes as we approach the end of the trading horizon, observe that  
\begin{equation}\label{asymptotic-covariance}
C_{[t_k, t_{k+1}]}^{i,j} \sim (t_{k+1}-t_{k})\Sigma_{i,j}, \quad \mbox{ as } t_{k+1} \to T,
\end{equation}
which means that one converges to market fundamentals at the end of the trading period.
This is due to the fact that, in our model, all traders have high enough risk aversions so that their trading speeds go to zero close to the terminal time $T$.

By virtue of \eqref{asset-returns-covariation}, one can also explain the realized correlation matrix in terms of the fundamental correlations $\rho_{i,j}:= \Sigma_{i,j}/(\Sigma_{i,i}\Sigma_{j,j})^{1/2}$. Namely, it holds that:
\begin{eqnarray} \label{asset-returns-correlation}
\\ \nonumber
R_{[t_k, t_{k+1}]}^{i,j} &=& \rho_{i,j} \left(  \frac{(t_{k+1}-t_{k})^2 \Sigma_{i,i}\Sigma_{j,j}}{C_{[t_k, t_{k+1}]}^{i,i}C_{[t_k, t_{k+1}]}^{j,j}}   \right)^{1/2} + \frac{\a_i \a_j \eta_i \eta_j \Lambda_k^{i,j} }{4V_i V_j\left(C_{[t_k, t_{k+1}]}^{i,i}C_{[t_k, t_{k+1}]}^{j,j} \right)^{1/2}} + \e_N\\ \nonumber
&=:& \rho_{i,j} A_k^{i,j} + B_k^{i,j} + \e_N
\end{eqnarray}
for any $1\leq i < j \leq d$. This expression shows that the deviation of the realized correlation from fundamentals is a linear map. The numerator of the multiplicative part $A_k^{i,j}$ does not depend on the off-diagonal terms of $\mathbb H$ while it is the case for the additive part $B_k^{i,j}$.
\subsection{Numerical Simulations} \label{COV-sect2}
In this part, we conduct several numerical experiments in order to illustrate the impact of trading on the structure of the covariance matrix of asset returns.

\begin{figure}[thbp] 
\begin{center}
        \subfigure[Intraday correlation]{%
            \label{figure1-realized-intraday-correlation}
            \includegraphics[width=0.5\textwidth]{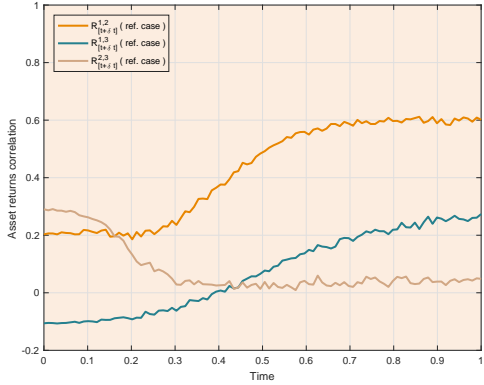}
        }%
        \subfigure[Intraday volatility for $\la=10^{4}$ and $\la=6.10^{3}$]{%
           \label{figure2-sensitivity-E_0}
           \includegraphics[width=0.5\textwidth]{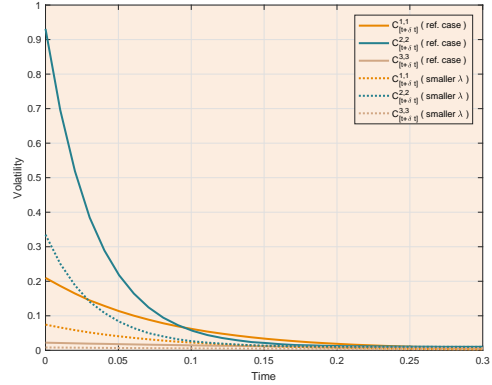} 
        }\\ 
        
        \subfigure[Intraday correlation for $\la=10^{3}$]{%
            \label{figure3-realized-intraday-correlation}
            \includegraphics[width=0.5\textwidth]{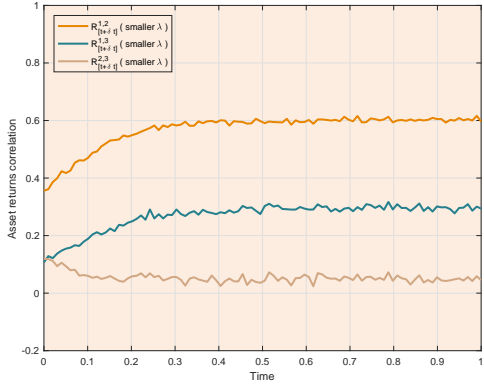}
        }%
         \subfigure[Intraday covariance for $\la=10^{4}$ and $\la=6.10^{3}$]{%
           \label{figure4-sensitivity-beta}
           \includegraphics[width=0.5\textwidth]{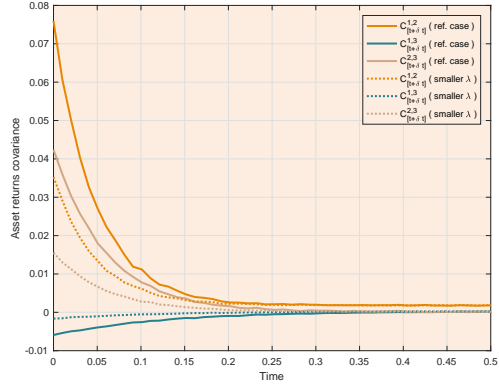} 
        }
        \end{center}
          \caption{%
        Simulated examples of intraday covariance and correlation matrices using \eqref{price-dynamics}.}
        \end{figure}    

We consider the example of Section \ref{Section-num-SF} by choosing $\rho_{1,2}=60 \%$, $\rho_{1,3}=30 \%$ and $\rho_{2,3}=5\%$. For  simplicity, we suppose that $E_0$ is a centered Gaussian random vector with a covariance matrix $\Gamma$ that is given by:
$$
\Gamma := \lambda^2.
 \left(
\begin{array}{ccc}
1 & 0.2 & -0.1 \\
0.2 & 1 & 0.3 \\
-0.1 & 0.3 & 1
\end{array}
\right),
$$
where $\lambda =10,000 \ share$. We fix a time step $\d t = 10^{-2} \ day$ ($\sim 4 \ min$), set $t_{k+1}-t_k = \d t$, and estimate $\left(C_{[t_k,t_k+\d t]}^{i,j} \right)_{\substack{1\leq i\leq j \leq 3 \\ 1\leq k \leq M-1}}$ and $\left(R_{[t_k,t_k+\d t]}^{i,j} \right)_{\substack{1\leq i<j \leq 3 \\ 1\leq k \leq M-1}}$ by generating a sample of $N=10,000$ observations using the numerical method of Section \ref{Section-num-SF}.

Figures \ref{figure1-realized-intraday-correlation}-\ref{figure4-sensitivity-beta} show that the observed covariance and correlation matrices are significantly deviated from fundamentals and especially at the beginning of the trading day. Figures \ref{figure2-sensitivity-E_0}-\ref{figure4-sensitivity-beta} also illustrates the sensitivity of the deviation relative to the change of the standard deviation of initial inventories: as $\lambda$ diminishes, the impact of trading is lower and the covariance and correlation matrices converge toward fundamentals.
 
On the other hand, we observe that the beginning of the trading period is dominated by the dependence structure of initial inventories. This is due to the domination of the additive terms $(B_k^{i,j})_{1\leq i<j \leq 3}$; in fact, given the relative high magnitude of denominator terms, $(A_k^{i,j})_{1\leq i<j \leq 3}$ are very small when $t_k \to 0$. Furthermore, we note that all the observed quantities converge toward fundamentals at the end of the trading period, which is in line with \eqref{asymptotic-covariance}.
 
  \subsection{Empirical Application}\label{section33-Empirical}
 
  Now, we carry out an empirical analysis on a pool of $d=176$ stocks. The data consists of five-minute binned trades ($\d t = 5$ min) and quotes information from January 2014 to December 2014, extracted from the primary market of each stock (NYSE or NASDAQ). \emph{We only focus on {the beginning of} the continuous trading session removing  30 min after the open and the last {90} min before the close,} in order to avoid  the particularities of trading activity in these periods {and target close strategies}.  Thus, the number of days is $N=252$ and the number of bins per day is $M={55}$. Days will be labelled by $l =1,...,N$, bins by $k = 1,..., M$, and for simplicity we note $C_k^{i,j}$ instead of $C_{[t_k-\d t, t_k]}^{i,j}$ for any $1\leq i, j \leq d$.
  
 {Our goal is to empirically assess the the influence of trading activity on the intraday covariance matrix of asset returns, and then compare the obtained models with our previous theoretical observations.} Given our analysis in Sections \ref{COV-sect1} and \ref{COV-sect2}, we expect an excess in the observed covariance matrix of asset returns and especially at the beginning of the trading period. Moreover, we expect the magnitude of this effect to be an increasing function of the typical size of market orders as it is noticed in Figures \ref{figure2-sensitivity-E_0} and \ref{figure4-sensitivity-beta}. 


\subsubsection{Market Impact} Let us start by assessing the relationship between the intraday {variance} of asset returns and the intraday {variance} of net exchanged flows $(F_k^{i,i})_{1\leq i\leq d}$, that is defined by:
$$
F_k^{i,i} := \frac{1}{N-1}  \sum_{l=1}^N \left( \nu_{k,l}^i -  \bar{\nu}_{k}^{i}   \right) 
\left( \nu_{k,l}^i -  \bar{\nu}_{k}^{i}   \right)
$$
for any $1\leq i \leq d$ and $k=1,...,M$; where $\nu_{k,l}^i$ is the net sum of exchanged volumes between $t_k- \d t$ and $t_k$ for stock $i$ in day $l$, and $\bar{\nu}_{k}^{i} = N^{-1} \sum_{l=1}^N \nu_{k,l}^i$ (i.e. $\bar{\nu}_{k}^{i}$ is an estimate of the expectation of $\nu_{k,l}^i$ regardless of the day).
{ As a by-product, we obtain estimates for the permanent market impact factors}.  Though $\nu$ does not represent exactly the same quantity as the variable $\mu$ of Section \ref{COV-sect1}, both variables are an indicator of market order flows {and for simplicity we shall use $\nu$ as a proxy for $\mu$.}

\begin{figure}[thbp] 
\begin{center}
        \subfigure[GOOG]{%
            \label{figure1-GOOG}
            \includegraphics[width=0.5\textwidth]{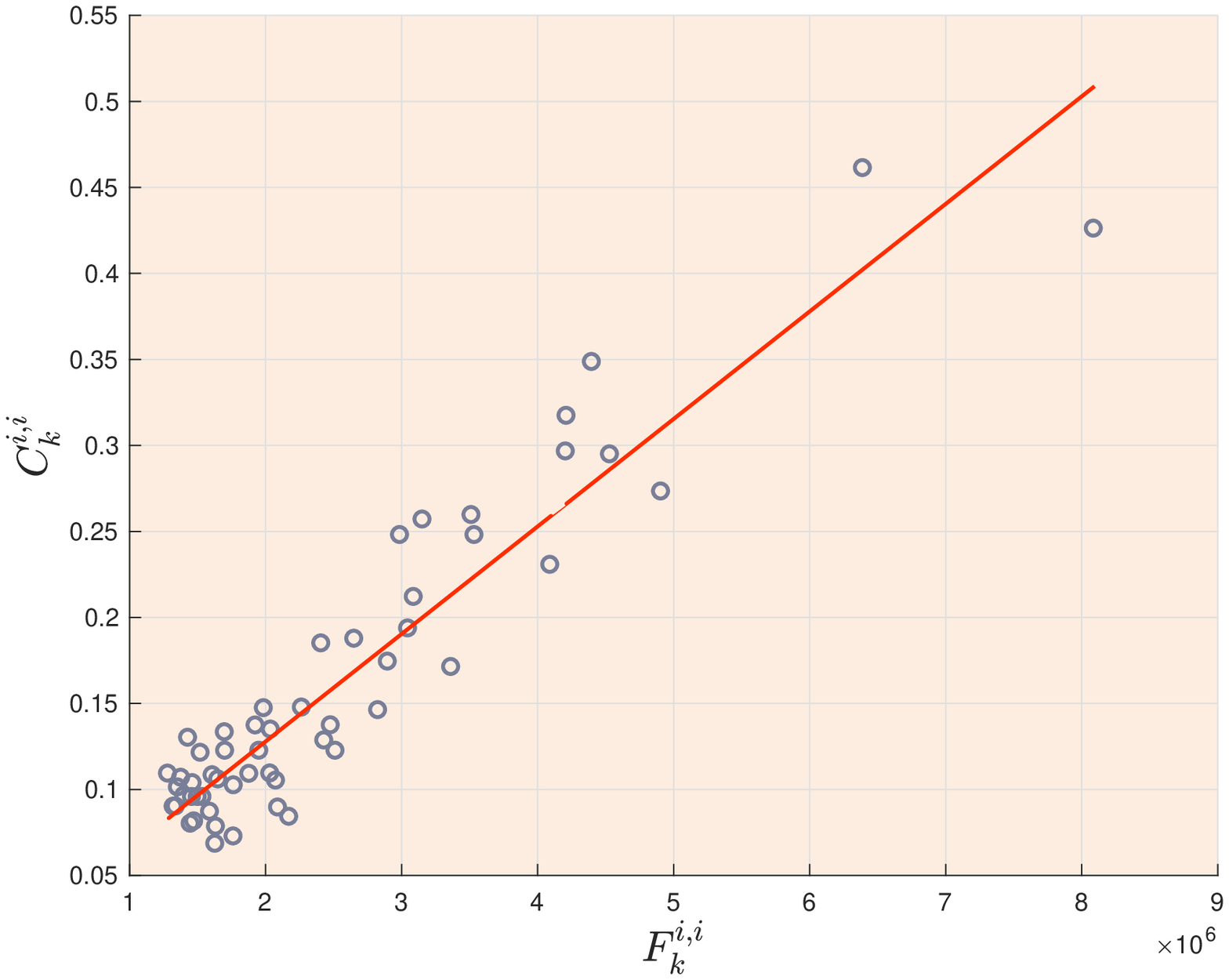}
        }%
        \subfigure[Histogram of ${\rm Corr} \left(C,F \right)$]{%
           \label{figure2-GOOG/AAPL}
           \includegraphics[width=0.5\textwidth]{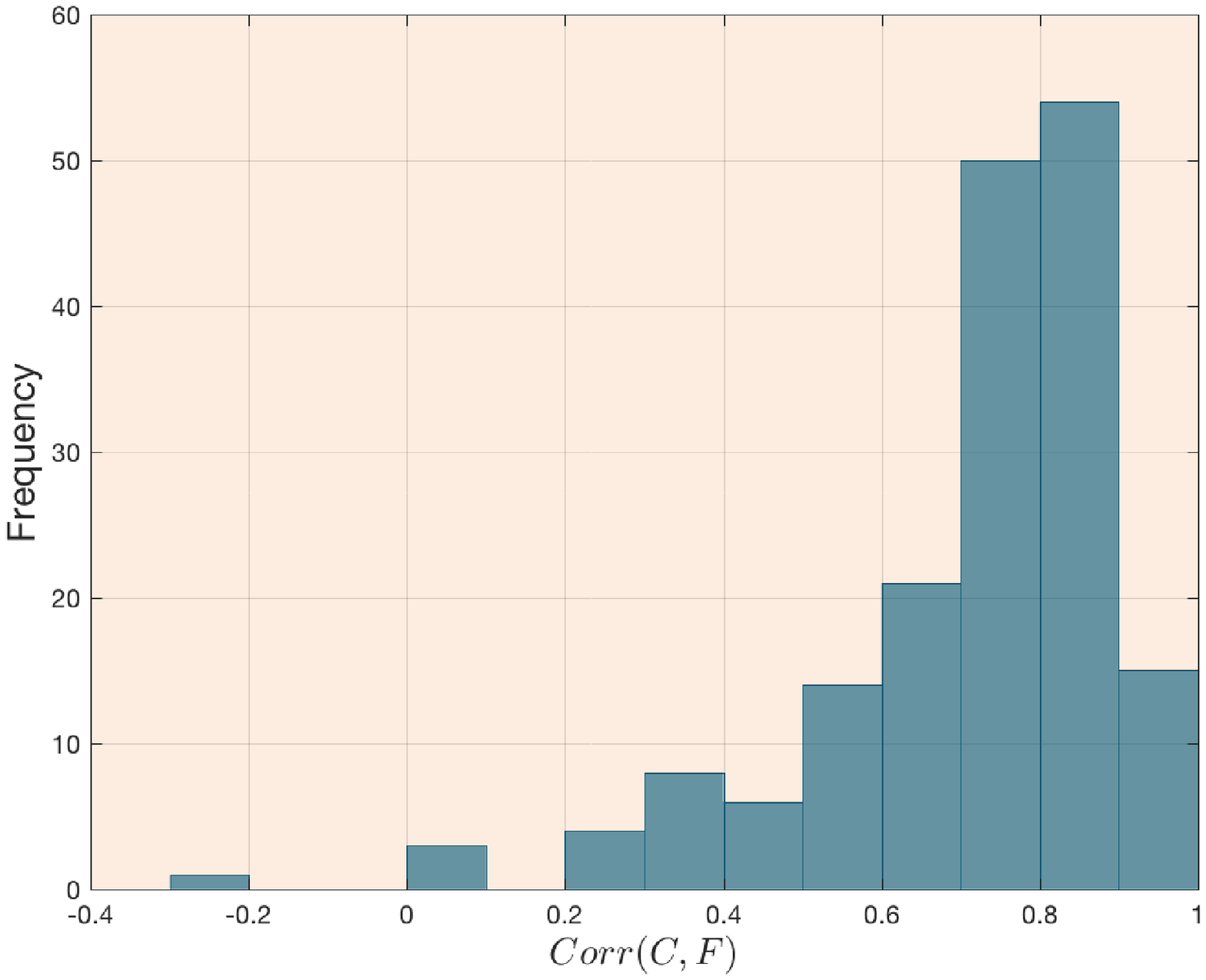} 
        }\\ 
   
        \end{center}
          \caption{%
         {Dependence structure between $(C_k^{i,i})_{1\leq k \leq M}$ and $(F_k^{i,i})_{1\leq k \leq M}$. Figure \ref{figure1-GOOG} displays the relationship between $(C_k^{i,i})_{1\leq k \leq M}$ and $(F_k^{i,i})_{1\leq k \leq M}$ for $GOOG$. Figure \ref{figure2-GOOG/AAPL} exhibits the histogram of correlations denoted ${\rm Corr}\left(C,F \right)$.}}
        \end{figure}    

{Figure \ref{figure1-GOOG}  shows a strong positive correlation between $(C_k^{i,i})_{1\leq k \leq M}$ and $(F_k^{i,i})_{1\leq k \leq M}$ for $GOOG$. Figure \ref{figure2-GOOG/AAPL} shows that this is true for almost all the stocks} and reinforces our findings in Sections \ref{COV-sect1} and \ref{COV-sect2}. Furthermore, as  \eqref{fundamental-relation-covariance} suggests, we suppose a linear relationship between $(C_k^{i,i})_{1\leq k \leq M}$ and $(F_k^{i,i})_{1\leq k \leq M}$; thus for every $1\leq i \leq d$ we fit the following regression:
\begin{equation}\label{market_impact_estimation_model}
C^{i,i} = \e  + \d t\cdot \Sigma + \a^2 \cdot F^{i,i}
\end{equation}
where $\e$ is the error term (assumed normal), the coefficient $\Sigma$ is related to the ``fundamental" covariance matrix of asset returns and the square root of the coefficient $\a^2$  is homogeneous to the market impact factor (cf. \eqref{fundamental-relation-covariance}). In Table \ref{table1-numbers} we show  estimates of $\a$, $\Sigma$ and the correlation between $(C_k^{i,i})_{1\leq k \leq M}$ and $(F_k^{i,i})_{1\leq k \leq M}$ (denoted ${\rm Corr}\left(C,F \right)$) for several examples. In particular, we obtain estimates for the permanent market impact  $\widehat{\a}$.

\begin{table}[thbp] 
\begin{center}
\setlength{\tabcolsep}{0.5em} 
{\renewcommand{\arraystretch}{1.5}%
\begin{tabular}{l l l l l l l}
  \hline 
    &  AAPL &  BMRN & GOOG & INTC  \\
  \hlineB{3} 
  $\widehat{\a}\ (bp) $&  $\bf 0.8\quad \ \ $ &  $\bf 8.43\quad \ \ $ & $\bf 2.5 \quad \ \ $ & $\bf 0.01$    \\
  $\widehat{\a^2} $&  $6.41 \times 10^{-11}\quad \ \ $ &  $7.11 \times 10^{-9}\quad \ \ $ & $6.25 \times 10^{-8} \quad \ \ $ & $1.79 \times 10^{-12}$    \\
   std.  &     $(4.15 \times 10^{-12}) $    &  $(3.98 \times 10^{-10})$ & $(3.17 \times 10^{-9}) $     & $(1.58 \times 10^{-13})$  \\
    p-value &  $0.01\%$ & $0.01\%$ & $0.01 \%$ &  $0.01\%$ \\\hline
  $ \widehat\Sigma$ &  $\bf 0.16$ & ${ -0.01}$ & ${ 0.15}$ &  $ 5.5 \times 10^{-3}$ \\
     std.  &   $(0.05)$    & $(0.05)$ & $(0.49)$      & $(2 \times 10^{-4})$ \\
    p-value    & $0.01\%$ & $60\%$ &  $75\%$ &  $2\%$ \\\hline
 ${\rm Corr}\left(C,F\right) \quad \ $ &  $90\%$ & $92\%$ & $94\%$ & $84\%$ \\      
  \hline
   \end{tabular}
   }
   \end{center}
   \vspace{0.2cm}
   \caption{%
        Estimates for $\a$, $\Sigma$ and the realized correlation between $(C_k^{i,i})_{1\leq k \leq M}$ and $(F_k^{i,i})_{1\leq k \leq M}$ for Nasdaq stocks. For each estimate the standard deviation (std.) is shown in parentheses and the p-value is given in the third row. Numbers in {\bf bold} are significant at a level of at least 99\%.}
        \label{table1-numbers}
\end{table}

\subsubsection{The Typical Intraday Pattern}
Next, we are interested in the intraday evolution of the diagonal and off-diagonal terms of the covariance matrix of returns, and in the way this evolution is affected when the typical size of trades diminishes. For that purpose, we compute the intraday covariance matrix of returns for our pool of US stocks and we normalize each term $(C_k^{i,j})_{1\leq k \leq M}$ by its daily average, then we consider the median value of diagonal terms and off-diagonal terms as a way of characterizing  the evolution of a typical diagonal term and a typical off-diagonal term respectively. The impact of the relative size of orders on the intraday patterns is assessed by conditioning our estimations.

More exactly, we start by defining the matrix of \emph{trade imbalances} for each stock $n$ in order to be able to compare the relative size of trades. Namely, for any $n,k,l$, we set:
$$
w_{k,l}^n:= \frac{\nu_{k,l}^n}{\underset{1\leq l \leq N}{mean} \sum_{k} | \nu_{k,l}^n|},
$$
where  $mean_{n\in A}\{ x_n\}$ denotes the average of $(x_n)$ as $n$ varies in $A$.
This mean is an estimate of the expectation of the sum of the absolute values of $\nu_{k,l}^n$ over a day; it can be seen as a renormalizing constant, enabling us to mix  different stocks on Figure \ref{Figure4-Intrinsec}.

Next, we define the \emph{conditioned intraday covariance matrix} $\left(C_k^{i,j}(\la) \right)_{\substack{1\leq i,j \leq d \\ 1\leq k \leq M}}$ for every $\la \geq 0$ as follows:
\begin{equation}\label{conditioned-covariance-chap5}
C_{k}^{i,j}(\la) := \frac{1}{\# \mathcal E_k^{i,j}(\la)-1} \sum_{l\in \mathcal E_k^{i,j}(\la)} \left( \d S^{i,k,l} -  \overline{\d S}_\la^{i,j,k}   \right)  
\left( \d S^{j,k,l} -  \overline{\d S}_\la^{j,i,k}    \right),
\end{equation}
where: 
\begin{itemize}
\item the set $\mathcal E_k^{i,j}(\la) $ corresponds to a conditioning: it contains only days for which this 5 min bins (indexed by $k$) for this pair of stocks (indexed by $(i,j)$, note that we can have $i=j$) is such that the renormalized net volumes are (in absolute value) below $\la$. It is strictly defined as follows:
$$
\mathcal E_k^{i,j}(\la) := \left\{ 1\leq l \leq N \ : \   |w_{k,l}^i | \leq \la \ \mbox{ and } \ |w_{k,l}^j | \leq \la  \right\};
$$
\item $\d S^{n,k, l}$  is the price increment defined as for \eqref{covariance-standard-estimation} and is computed from the historic stock prices;  
\item $\overline{\d S}_\la^{i,j,k}$ is the average price increment over selected days, given by: $\overline{\d S}_\la^{i,j,k}= \left( \sum _{l\in \mathcal E_k^{i,j}(\la)} \d S^{i,k, l}\right)/ \left(\# \mathcal E_k^{i,j}(\la) \right)$;
\item $\# \mathcal E_k^{i,j}(\la) $ denotes the number of elements of $\mathcal E_k^{i,j}(\la)$: the number of selected days. Note that the stricter the conditioning (i.e. the smaller $\la$), the less days in the selection, and hence the smaller $\# \mathcal E_k^{i,j}(\la) $.
\end{itemize}
Here $\left(C_k^{i,j}(\la) \right)_{\substack{1\leq i,j \leq d \\ 1\leq k \leq M}}$ represents the intraday covariance matrix of returns conditioned on trade imbalances between $-\la$ and $\la$. In all our examples, the coefficient $\la$ is chosen to have enough days in the selection (for obvious statistical significance reasons), i.e. so that $\# \mathcal E_k^{i,j}(\la) \gg 1$ for any $1\leq i,j \leq d$ and $1\leq k \leq M$. 
\begin{figure}[thbp] 
\begin{center}
        \subfigure[The median diagonal pattern (squared volatility)]{%
            \label{figure1-empirical-study}
            \includegraphics[width=0.5\textwidth]{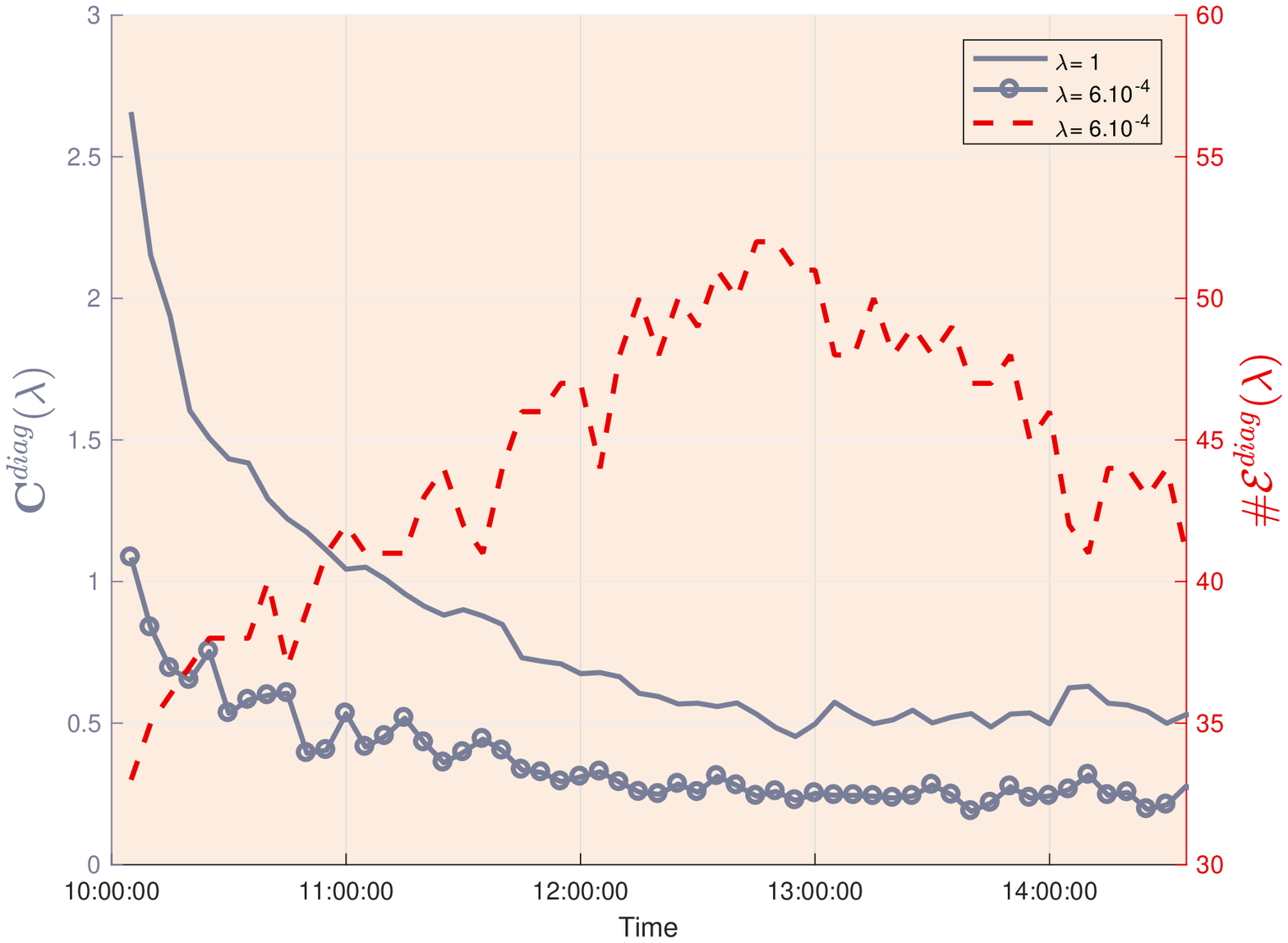}
        }%
        \subfigure[The median off-diagonal pattern]{%
           \label{figure2-empirical-study}
           \includegraphics[width=0.5\textwidth]{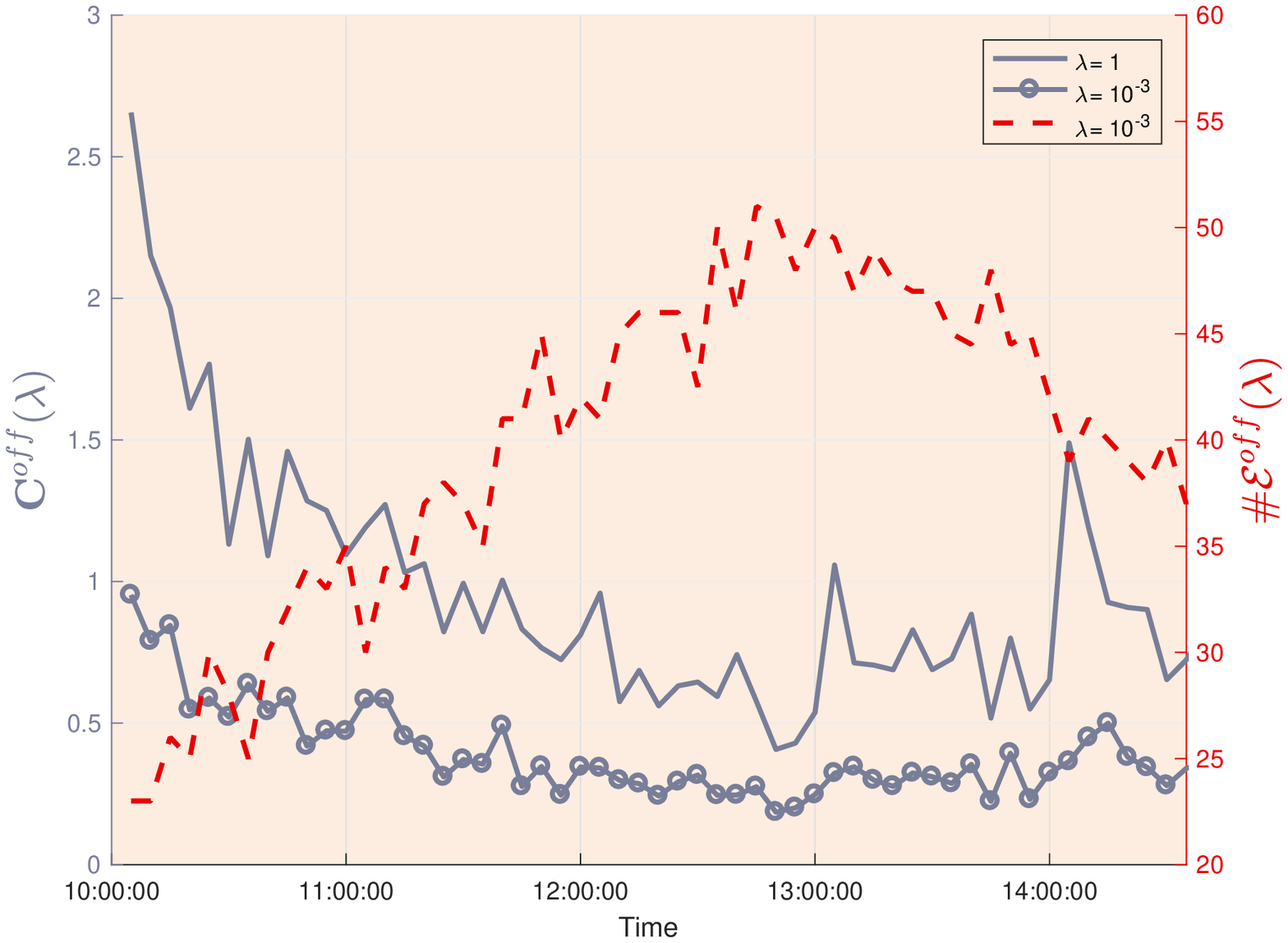} 
        }%
        \end{center}
        
          \caption{%
          Plots of the median diagonal pattern ${\bf C}^{diag}(\la)$ and the median off-diagonal pattern ${\bf C}^{off}(\la)$ for diverse values of $\la$.
          { The secondary axis corresponds to the number of observations for each 5 minutes bin after the conditioning.}}
          \label{Figure4-Intrinsec}
        \end{figure}

Now we define the median diagonal pattern ${\bf C}^{diag}(\la):=\left(C_k^{diag}(\la) \right)_{1\leq k \leq M}$ and the median off-diagonal pattern ${\bf C}^{off}(\la):=\left(C_k^{off}(\la) \right)_{1\leq k \leq M}$ as follows: 
$$
C_k^{diag}(\la) := \underset{1\leq i \leq d}{median}\left\{  C_k^{i,i}(\la) / \underset{1\leq k \leq M}{mean} \left\{ C_k^{i,i}(1) \right\}   \right\}
$$
and
$$
 C_k^{off}(\la) := \underset{1\leq i<j \leq d}{median}\left\{  C_k^{i,j}(\la) / \underset{1\leq k \leq M}{mean} \left\{ C_k^{i,j}(1) \right\}   \right\},
$$
for any $k=1,...,M$ and  $\la \geq 0$. Here the notation $median_{n\in A}\{x_n\}$  denotes the median  value of $(x_n)$ as $n$ varies in $A$. One should note that the choice of the normalization constant (i.e. the mean over bins of $C_k^{i,j}(1)$) will allow us to compare the different curves with respect to the reference case, i.e. without conditioning. In fact, it turns out that $\la=1$ removes all conditionings: $1$ is above the maximum value of our renormalized flows. {Moreover, we set  $\# {\bf {\mathcal{E}}}_k^{diag}(\la) := \underset{1\leq i \leq d}{median}\left\{  \#{\mathcal{E}}_k^{i,i}(\la)\right\}$ and  $\# {\bf {\mathcal{E}}}_k^{off}(\la) := \underset{1\leq i<j \leq d}{median}\left\{  \#{\mathcal{E}}_k^{i,j}(\la)\right\}$.}

We take medians instead of means to have robust estimates of the expectations. We do not want our estimates to be polluted by few days of potential erratic market data, that could for instance be due to trading interruptions.

Figures \ref{figure1-empirical-study} and \ref{figure2-empirical-study} displays representations of ${\bf C}^{diag}(\la)$ and ${\bf C}^{off}(\la)$ for various values of $\la$. Observe that  ${\bf C}^{diag}(1)$ and ${\bf C}^{off}(1)$ exhibits a pattern that is very similar to our simulation in Figures \ref{figure2-sensitivity-E_0} and \ref{figure4-sensitivity-beta}, especially between the beginning of the trading period and $13:00$. Indeed, the observed quantities are high at the beginning of the trading period, lower as the day progresses until it reaches a minimum around $13:00$, followed by a slight increase until market close. The general shape of these curves (left-slanted smile) is well-known (see e.g. \cite{book-Jaimungal} and references therein).  


Our core observation is that: \emph{given the absolute value of the net flows are small, this average curve {flattens out}, even at the beginning of the day}. At our knowledge, it is the first time that this conditioning is mentioned, and it is perfectly in line with our simulated  Figures \ref{figure2-sensitivity-E_0} and \ref{figure4-sensitivity-beta}. This suggests the slopes of the ``averaged volatility curves'' comes essentially from the days during which there is a large positive or negative imbalance of large orders, that are ``optimally'' executed. { We believe that this analysis  should be completed by using a larger data set.}

\subsubsection{A Toy Model Calibration}
Now, we use historical data to fit our MFG model to some examples of traded stocks. For that purpose, we use a very simple approach by reducing as much as possible the number of parameters: 
\begin{enumerate}[ label=($\mathcal{S}$\arabic*)]
\item\label{step1} We suppose that $E_0$ is a centered Gaussian random vector with a covariance matrix $\Gamma$. Moreover, as suggested by \eqref{E-derivative-mu} , we use ${\rm Corr}\left( \sum_{k} \nu_{k}^i, \sum_{k} \nu_{k}^j \right)$ as a proxy for ${\rm Corr}\left( E_0^i, E_0^j \right)$, and which is in turn estimated from data by using the standard estimator :
$$
\frac{1}{N-1} \sum_{l=1}^N \left(\sum_{k} \nu_{k,l}^{i}- \overline{\sum_{k} \nu_{k,l}^{i}} \right) \left(\sum_{k} \nu_{k,l}^{j} - \overline{\sum_{k} \nu_{k,l}^{j}} \right),
$$
where $\overline{\sum_{k} \nu_{k,l}^{i}} = N^{-1} \sum_{l}\sum_{k} \nu_{k,l}^{i}$.

\item\label{step2} As suggested by Figures \ref{figure1-empirical-study}-\ref{figure2-empirical-study} we choose
$$
\d t \Sigma_{i,j} = 0.2 \times \underset{1\leq k \leq M}{mean} \left\{ C_k^{i,j}(1) \right\},
$$
and we shift upward the simulated curves by $\d = 0.3 \times \underset{1\leq k \leq M}{mean} \left\{ C_k^{i,j}(1) \right\}$;
\item Finally, we fix the penalization parameters $A_i=A=10$, and choose $k_i := V_i/\eta_i$, $\gamma$, and $\Gamma_{i,i}$ by minimizing the $L^2$-error between the simulated curves and real curves.
\end{enumerate}

\begin{figure}[thbp] 
\begin{center}
        \subfigure[GOOG]{%
            \label{figure1-GOOG-fit}
            \includegraphics[width=0.5\textwidth]{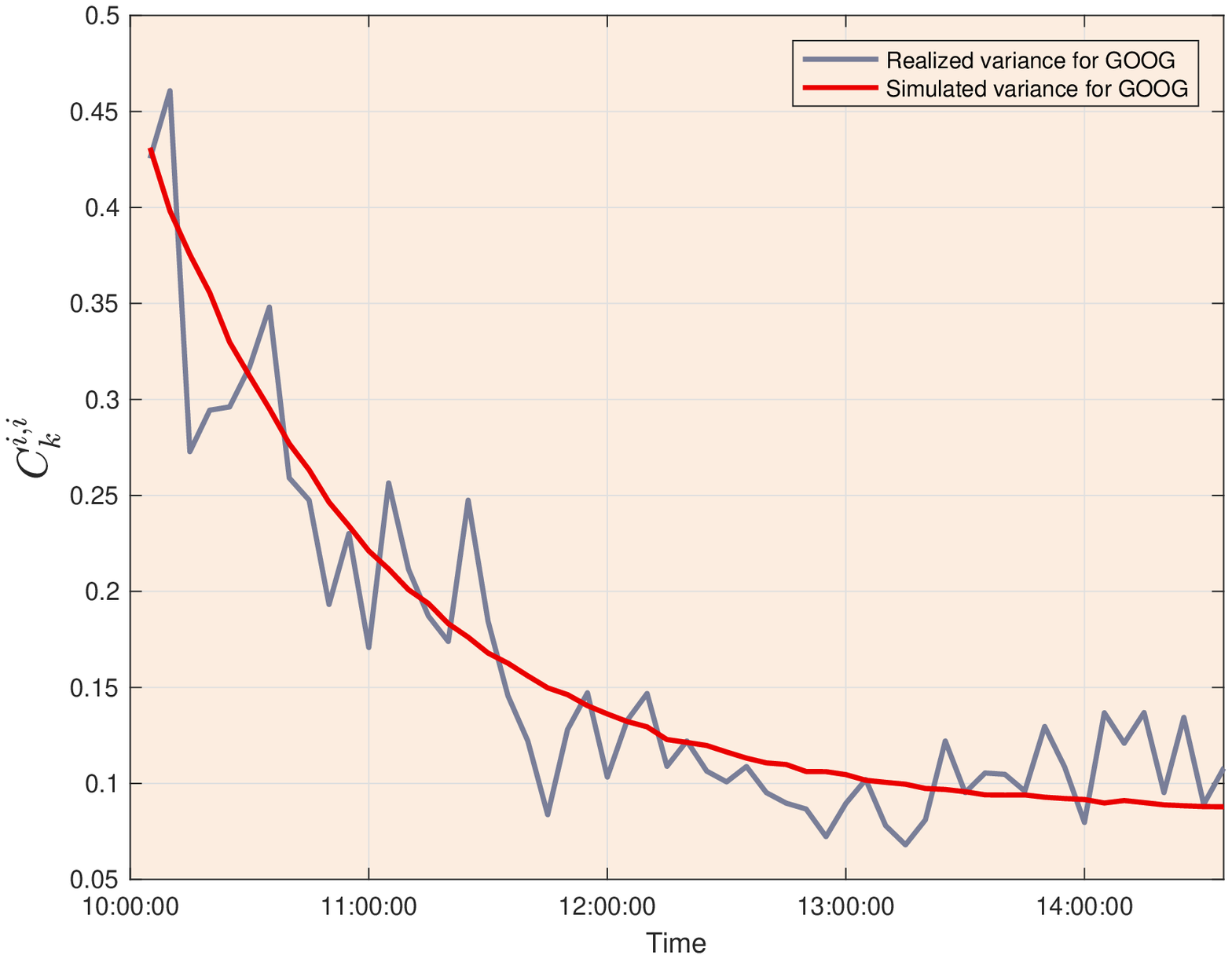}
        }%
        \subfigure[GOOG/AAPL]{%
           \label{figure2-GOOG/AAPL-fit}
           \includegraphics[width=0.5\textwidth]{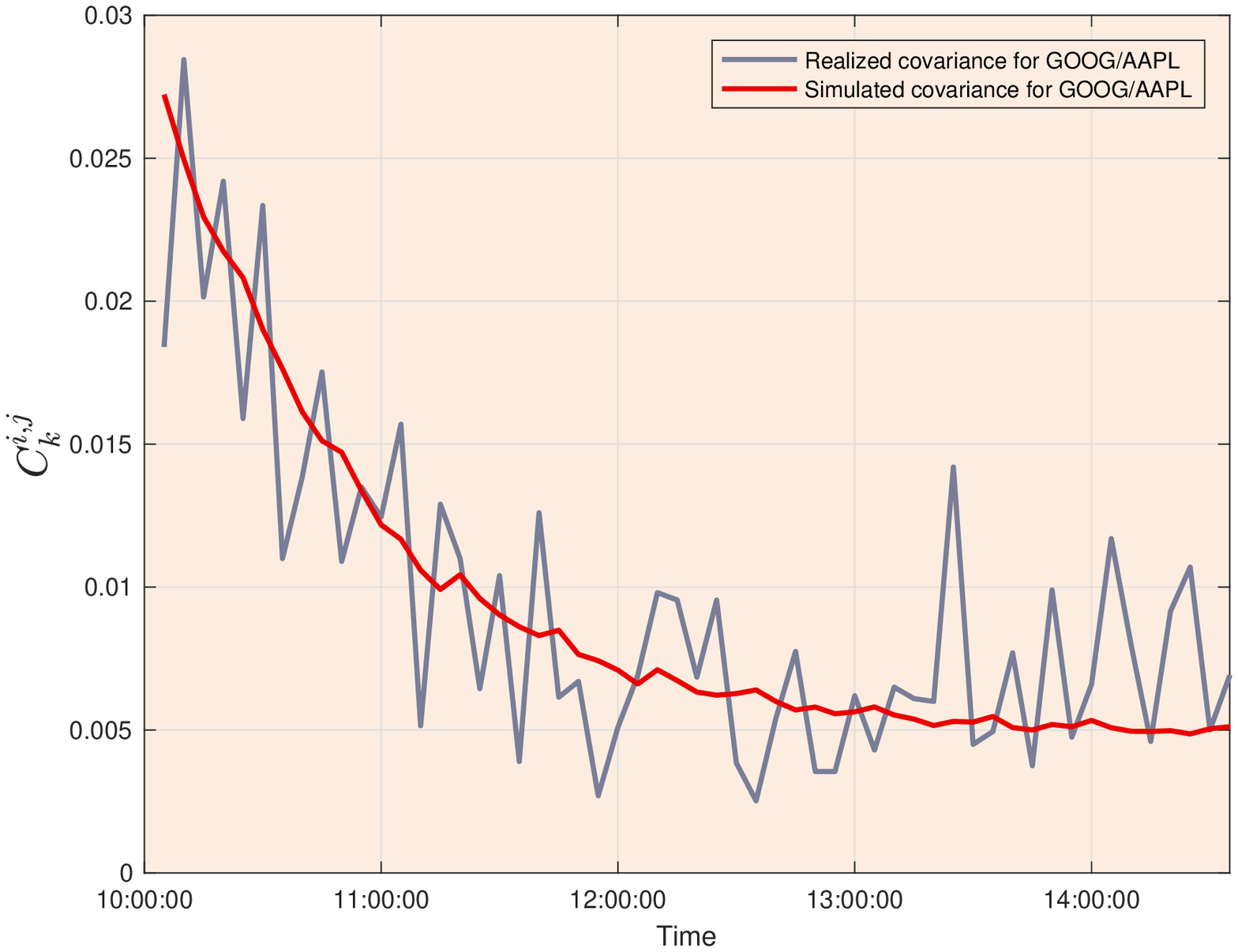} 
        }\\ 
   
        \end{center}
          \caption{%
         Comparison between the simulated curves and the real curves for two examples. Figure \ref{figure1-GOOG-fit}  corresponds to $i \equiv j \equiv GOOG$ and Figure \ref{figure2-GOOG/AAPL-fit} corresponds to $(i,j)\equiv (GOOG, AAPL)$.}
        \end{figure}

\begin{table}[thbp] 
\begin{center}
\setlength{\tabcolsep}{0.5em} 
{\renewcommand{\arraystretch}{1.5}%
  \begin{tabular}{l l l l l}\hline
    \multicolumn{3}{l}{\bf Estimated using the regression~\eqref{market_impact_estimation_model} of Section~\ref{section33-Empirical} and \ref{step1}-\ref{step2}}\\
    ${\rm Corr}\left( E_0^1, E_0^2 \right)= 20 \%$ , & $\a_1=2.5\times 10^{-4}$, & $\sigma_1 = 1.55 $,   \\
    $\rho_{1,2}=0.5 \%$, & $\a_2=7.9\times 10^{-5} $, & $\sigma_2=0.43 $,    \\\hline
    \multicolumn{3}{l}{\bf Calibrated on curves of Figure \ref{figure1-GOOG-fit} and \ref{figure2-GOOG/AAPL-fit}}\\
    $\Gamma_{1,1}= 3.6 \times 10^{9} $, &$\Gamma_{2,2}=2.02\times 10^{9}$, & $\g=10^{-3} $, \\
                       $k_1= 2\times 10^{7} $,   & $k_2= 8\times 10^{8}$. &  \\[0.3cm]
   \end{tabular}
   }
   \end{center}
   \vspace{0.2cm}
   \caption{%
        The MFG model parameters for the two-stocks portfolio:  Asset 1 $\equiv GOOG$; Asset 2 $\equiv AAPL$.}
        \label{table2-calibration}
\end{table}

Figures \ref{figure1-GOOG-fit}-\ref{figure2-GOOG/AAPL-fit} show illustrative examples by considering the two-stocks portfolio:  Asset 1 $\equiv GOOG$; Asset 2 $\equiv AAPL$. For that example, {the parameters of our model} are presented in Table \ref{table2-calibration}.

{Here $\Gamma_{1,2}$, $\a_1$, $\a_2$, $\sigma_1$, $\sigma_2$, $\rho_{1,2}$ are estimated from data (cf. Table \ref{table1-numbers} and Figures \ref{figure1-empirical-study}-\ref{figure2-empirical-study}), while $\Gamma_{1,1}$, $\Gamma_{2,2}$, $\g$, $k_1$, $k_2$ are computed by minimizing the $L^2$-error between the simulated curves and real curves. Following this approach,  one requires $2d+1$ parameters to fit a portfolio of $d$ stocks (i.e. $d(d+1)/2$ curves). }

\newpage
\bibliographystyle{siam}
\bibliography{C:/mybib/mybib}

\begin{bibdiv}
\begin{biblist}

\bib{alfonsi2010optimal}{article}{
title={Optimal execution strategies in limit order books with general shape functions},
  author={Alfonsi, A.}
  author={ Fruth, A.}
  author={ Schied, A.},
  journal={Quantitative Finance},
  volume={10},
  number={2},
  pages={143--157},
  year={2010},
  publisher={Taylor \& Francis}
}

\bib{alfonsi2010optimal-bis}{article}{
title={Optimal trade execution and absence of price manipulations in limit order book models},
  author={Alfonsi, A.}
  author={Schied, A.},
  journal={SIAM Journal on Financial Mathematics},
  volume={1},
  number={1},
  pages={490--522},
  year={2010},
  publisher={SIAM}
  }


\bib{almgren2003optimal}{article}{
title={Optimal execution with nonlinear impact functions and trading-enhanced risk},
  author={Almgren, R.},
  journal={Applied mathematical finance},
  volume={10},
  number={1},
  pages={1--18},
  year={2003},
  publisher={Taylor \& Francis}
}

\bib{almgren2012optimal}{article}{
title={Optimal trading with stochastic liquidity and volatility},
  author={Almgren, R.},
  journal={SIAM Journal on Financial Mathematics},
  volume={3},
  number={1},
  pages={163--181},
  year={2012},
  publisher={SIAM}
}

\bib{almgren1999value}{article}{
 title={Value under liquidation},
  author={Almgren, R.} 
  author={Chriss, N.},
  journal={Risk},
  volume={12},
  number={12},
  pages={61--63},
  year={1999}
}

\bib{almgren2001optimal}{article}{
      author={Almgren, R.},
      author={Chriss, N.},
       title={Optimal execution of portfolio transactions},
        date={2000},
     journal={Journal of Risk},
     volume={ 3(2)},
      pages={5-39},
}

\bib{almgren2005}{article}{
      author={Almgren, R.},
      author={Thum, C.},
      author={Hauptmann, E.},
      author={Li, H.},
       title={Direct Estimation of Equity Market Impact},
        date={2005},
     journal={Risk},
     volume={ 18},
      pages={57-62},
}

\bib{bayraktar2014liquidation}{article}{
  title={Liquidation in limit order books with controlled intensity},
  author={Bayraktar, E.}
  author= {Ludkovski, M.},
  journal={Mathematical Finance},
  volume={24},
  number={4},
  pages={627--650},
  year={2014},
  publisher={Wiley Online Library}
}

\bib{benzaquen2017dissecting}{article}{
author={Benzaquen, M.} 
author={Mastromatteo, I.} 
author={Eisler, Z.}
author={Bouchaud, J.-P.},
title={Dissecting cross-impact on stock markets: An empirical analysis},
  journal={Journal of Statistical Mechanics: Theory and Experiment},
  volume={2017},
  number={2},
  pages={023406},
  year={2017},
  publisher={IOP Publishing}
}

\bib{BLA98}{article}{
    author = {Bertsimas, Dimitris},
    author = {Lo, Andrew W.},
    journal = {Journal of Financial Markets},
    number = {1},
    pages = {1--50},
    title = {{Optimal control of execution costs}},
    volume = {1},
    year = {1998},
  }
  
\bib{boulatov2012informed}{article}{
  title={Informed trading and portfolio returns},
  author={Boulatov, A.}
  author= {Hendershott, T.} 
  author= {Livdan, D.},
  journal={Review of Economic Studies},
  volume={80},
  number={1},
  pages={35--72},
  year={2012},
  publisher={Oxford University Press}
}


\bib{book-Lehalle}{book}{
    editor = {Lehalle, Charles-Albert},
    editor = {Laruelle, Sophie},
   author={Burgot, R.},
   author={Lasnier, M.},
    author={Lehalle, C.-A},
    author={Laruelle, S.},
   publisher={World Scientific publishing},
    title = {{Market Microstructure in Practice (2nd edition)}},
    year = {2018}
}


\bib{cardaliaguet2016mfgcontrols}{article}{
      author={Cardaliaguet, Pierre},
      author={Lehalle, Charles-Albert},
       title={Mean field game of controls and an application to trade
  crowding},
        date={forthcoming},
     journal={Market Microstructure and Liquidity},
  publisher={World Scientific}
}

\bib{master-equation}{article}{
      author={Cardaliaguet, P.},
      author={Delarue, F.},
      author={Lasry, J.-M.},
      author={Lions, P.-L.},
       title={The master equation and the convergence problem in mean field games},
        date={2015},
     journal={arXiv:1509.02505v1},
}

\bib{cartea2015optimal}{article}{
        title={Optimal execution with limit and market orders},
  author={Cartea, A.}
  author={ Jaimungal, S.},
  journal={Quantitative Finance},
  volume={15},
  number={8},
  pages={1279--1291},
  year={2015},
  publisher={Taylor \& Francis}
}

\bib{order-flow2015}{article}{
      author={Cartea, A.},
      author={Jaimungal, S.},
       title={Incorporating Order-Flow into Optimal Execution},
        date={2016},
     journal={Math Finan Econ},
     volume={ 10},
      pages={339-364},
}
\bib{book-Jaimungal}{book}{
   author={Cartea, A.},
   author={Jaimungal, S.},
    author={Penalva, J.},
   title={Algorithmic and High-Frequency Trading},
   series={Mathematics, Finance and Risk},
   publisher={Cambridge University Press },
   date={2015},
}


\bib{cont2013running}{article}{
  title={Running for the exit: distressed selling and endogenous correlation in financial markets},
  author={Cont, Rama}, 
  author={Wagalath, Lakshithe}, 
  journal={Mathematical Finance: An International Journal of Mathematics, Statistics and Financial Economics},
  volume={23},
  number={4},
  pages={718--741},
  year={2013},
  publisher={Wiley Online Library}
}

\bib{curato2017optimal}{article}{
  title={Optimal execution with non-linear transient market impact},
  author={Curato, G.}
  author= {Gatheral, J.} 
  author={Lillo, F.},
  journal={Quantitative Finance},
  volume={17},
  number={1},
  pages={41--54},
  year={2017},
  publisher={Taylor \& Francis}}

\bib{caine16trade}{inproceedings}{
    author = {Firoozi, D.}, 
    author = {Caines, P. E.}, 
    booktitle = {2016 IEEE 55th Conference on Decision and Control (CDC)},
    pages = {268--275},
    title = {{Mean Field Game epsilon-Nash equilibria for partially observed optimal execution problems in finance}},
    year = {2016}
  }


\bib{gatheral2010no}{article}{
title={No-dynamic-arbitrage and market impact},
  author={Gatheral, J.},
  journal={Quantitative finance},
  volume={10},
  number={7},
  pages={749--759},
  year={2010},
  publisher={Taylor \& Francis}
}

\bib{gatheral2012transient}{article}{
title={Transient linear price impact and Fredholm integral equations},
  author={Gatheral, J.}
  author={ Schied, A.} 
  author={Slynko, A.},
  journal={Mathematical Finance: An International Journal of Mathematics, Statistics and Financial Economics},
  volume={22},
  number={3},
  pages={445--474},
  year={2012},
  publisher={Wiley Online Library}
}

\bib{Gueant}{book}{
   author={Olivier Gu\'eant },
   title={The Financial Mathematics of Market Liquidity: From Optimal Execution to Market Making},
   series={Chapman and Hall/CRC Financial Mathematics Series},
   publisher={Chapman and Hall/CRC },
   date={2016},
}

\bib{gueant2012optimal}{article}{
      title={Optimal portfolio liquidation with limit orders},
  author={Gu{\'e}ant, O.}
  author= {Lehalle, C.-A.}
  author= {Fernandez-Tapia, J.},
  journal={SIAM Journal on Financial Mathematics},
  volume={3},
  number={1},
  pages={740--764},
  year={2012},
  publisher={SIAM}
}

\bib{gueant2015convex}{article}{
  title={A convex duality method for optimal liquidation with participation constraints},
  author={Gu{\'e}ant, Olivier},
  author={Lasry, Jean-Michel},
  author={Pu, Jiang},
  journal={Market Microstructure and Liquidity},
  volume={1},
  number={01},
  pages={1550002},
  year={2015},
  publisher={World Scientific}
}

\bib{hasbrouck2001common}{article}{
      title={Common factors in prices, order flows, and liquidity},
  author={Hasbrouck, J.} 
  author={Seppi, D. J.},
  journal={Journal of financial Economics},
  volume={59},
  number={3},
  pages={383--411},
  year={2001},
  publisher={Elsevier}      } 

\bib{huitema2014optimal}{article}{
title={Optimal portfolio execution using market and limit orders},
  author={Huitema, R.},
  year={2014}
}


\bib{citeulike:13586166}{article}{
    author={Jaimungal, Sebastian}, 
    author={Nourian, Mojtaba}, 
    author={Huang, Xuancheng}, 
    day={16},
    journal={Social Science Research Network Working Paper Series},
    month={mar},
    title={{Mean-Field Game Strategies for a Major-Minor Agent Optimal Execution Problem}},
   year={2015}
 }


\bib{kuvcera1973review}{article}{
      author={Ku{\v{c}}era, V.},
       title={A review of the matrix Riccati equation},
  journal={Kybernetika},
  volume={9},
  number={1},
  pages={42--61},
  year={1973},
  publisher={Institute of Information Theory and Automation AS CR}
}
 
\bib{kratz2014optimal}{article}{
     title={Optimal liquidation in dark pools},
  author={Kratz, P.}
  author={Sch{\"o}neborn, T.},
  journal={Quantitative Finance},
  volume={14},
  number={9},
  pages={1519--1539},
  year={2014},
  publisher={Taylor \& Francis}
}


\bib{martin1968exponential}{article}{
title={On the exponential representation of solutions of linear differential equations},
  author={Martin, J.F.P.},
  journal={Journal of differential equations},
  volume={4},
  number={2},
  pages={257--279},
  year={1968},
  publisher={Academic Press}
}

\bib{mastromatteo2017trading}{article}{
title={Trading lightly: Cross-impact and optimal portfolio execution},
  author={Mastromatteo, I.}
  author= {Benzaquen, M.} 
  author={Eisler, Z.} 
  author= {Bouchaud, J.-P.},
  year={2017}
}

\bib{schied2017high}{article}{
  title={High-frequency limit of Nash equilibria in a market impact game with transient price impact},
  author={Schied, Alexander}, 
  author={Strehle, Elias}, 
  author={Zhang, Tao}, 
  journal={SIAM Journal on Financial Mathematics},
  volume={8},
  number={1},
  pages={589--634},
  year={2017},
  publisher={SIAM}
}

\end{biblist}
\end{bibdiv}
\end{document}